\documentclass[journal=jctcce,manuscript=article,
layout=twocolumn]{achemso}
\setkeys{acs}{articletitle = true}

\usepackage{mathtools}
\usepackage{tikz}

\usepackage[numbers]{natbib}
\usepackage{notoccite}
\usepackage[english]{babel}

\usepackage{amsmath}
\usepackage{amsthm}
\usepackage[colorlinks=true, allcolors=blue]{hyperref}
\usepackage{graphicx} % Include figure files
\usepackage{dcolumn} % Align table columns on decimal point
\usepackage{bm} % bold math
\usepackage{braket}
\usepackage{color}
\usepackage{amsthm}
\usepackage{hyperref} % add hypertext capabilities
\usepackage[normalem]{ulem} %to strike the words
\usepackage{qcircuit}       % professional-quality tables

\definecolor{lightblue}{RGB}{73,151,208}
\definecolor{crimson}{RGB}{140,41,53}

\hypersetup{
    colorlinks,
    linkcolor={crimson},
    citecolor={lightblue},
    urlcolor={lightblue}
}

\newtheorem{problem}{Problem}
\newtheorem{lemma}{Lemma}
\newtheorem{theorem}{Theorem}
\newtheorem{corollary}{Corollary}[theorem]
\newtheorem{remark}{Remark}

\def\>{\rangle}
\def\<{\langle}

\DeclareMathOperator*{\argmin}{argmin}

\newcommand{\CNOT}{\mathrm{CNOT}}
\newcommand{\SWAP}{\mathrm{SWAP}}

\DeclareTextFontCommand{\todo}{\bfseries\color{red}}
\usepackage{mathtools}
\DeclarePairedDelimiter{\abs}{\lvert}{\rvert}

\DeclarePairedDelimiter{\floor}{\lfloor}{\rfloor}

\newcolumntype{L}[1]{>{\raggedright\let\newline\\\arraybackslash\hspace{0pt}}m{#1}}
\newcolumntype{C}[1]{>{\centering\let\newline\\\arraybackslash\hspace{0pt}}m{#1}}
\newcolumntype{R}[1]{>{\raggedleft\let\newline\\\arraybackslash\hspace{0pt}}m{#1}}

\include{MyCommand}

\DeclareMathOperator{\wt}{wt}
\DeclareMathOperator{\im}{im}
\usepackage{verbatim}

\usepackage{algorithm, algpseudocodex}
\mciteErrorOnUnknownfalse

\author{Jeffery Yu}
\affiliation{Joint Center for Quantum Information and Computer Science, NIST/University of Maryland,
College Park, Maryland 20742, USA}
\alsoaffiliation{Joint Quantum Institute, NIST/University of Maryland,
College Park, Maryland 20742, USA}
\altaffiliation{Physics and Informatics Laboratory, NTT Research, Inc., Sunnyvale, California 94085, USA}

\author{Yuan Liu}
\affiliation{Department of Electrical and Computer Engineering, North Carolina State University, Raleigh, North Carolina 27606, USA}
\alsoaffiliation{Department of Computer Science, North Carolina State University, Raleigh, North Carolina 27606, USA}
\alsoaffiliation{Department of Physics, North Carolina State University, Raleigh, North Carolina 27606, USA}

\author{Sho Sugiura}
\affiliation{Blocq, Inc., 717 Market Street, San Francisco, California 94103, USA}
\altaffiliation{Physics and Informatics Laboratory, NTT Research, Inc., Sunnyvale, California 94085, USA}

\author{Troy Van Voorhis}
\affiliation{Department of Chemistry, Massachusetts Institute of Technology, Cambridge, Massachusetts 02139, USA}

\author{Sina Zeytino\u glu}
\affiliation{Institute for Theoretical Physics, TU Wien, Wiedner Hauptstraße 8-10/136, A-1040 Vienna, Austria}
\email{sina.zeytinoglu@tuwien.ac.at}
\altaffiliation{Physics and Informatics Laboratory, NTT Research, Inc., Sunnyvale, California 94085, USA}

\title{Clifford circuit based heuristic optimization of fermion-to-qubit mappings }

\begin{document}

%%%%%%%%%%%%%%%%%%%%%%%%%%

%%%%%%%%%%%%%%%%%%%%%%%%%%
\begin{abstract}
Simulation of interacting fermionic Hamiltonians is one of the most promising applications of quantum computers. However, the feasibility of analysing 
fermionic systems with a quantum computer hinges on the efficiency of 
fermion-to-qubit mappings that encode non-local fermionic degrees of freedom in local qubit degrees of freedom. While recent works have highlighted the importance of designing fermion-to-qubit mappings that are tailored to specific problem Hamiltonians, the methods proposed so far are either restricted to a narrow class of mappings or they use computationally expensive and unscalable brute-force search algorithms.   
Here, we address this challenge by designing a \textit{heuristic} numerical optimization framework for fermion-to-qubit mappings. To this end, we first translate the fermion-to-qubit mapping problem to a Clifford circuit optimization problem, and then use simulated annealing to optimize the average Pauli weight of the problem Hamiltonian. 
For all fermionic Hamiltonians we have considered, the numerically optimized mappings outperform their conventional counterparts, including ternary-tree-based mappings that are known to be optimal for single creation and annihilation operators. We find that our optimized mappings yield between 15\% to 40\% improvements on the average Pauli weight when the simulation Hamiltonian has an intermediate level of complexity. Most remarkably, the optimized mappings improve the average Pauli weight for $6 \times 6$ nearest-neighbor hopping and Hubbard models by more than $40\%$ and $20\%$, respectively. Surprisingly, we also find specific interaction Hamiltonians for which the optimized mapping outperform \textit{any} ternary-tree-based mapping. Our results establish heuristic numerical optimization as an effective method for obtaining mappings tailored for specific fermionic Hamiltonian.  
\end{abstract}

\section{Introduction}

Solving the many-electron Schr\"{o}dinger equation is of paramount importance for the scientific understanding of physical properties of molecules and materials \cite{Schiffer_PRL_1995,Jarrell_EPL_2001,bohn2017cold,balakrishnan2016perspective}. However, the complexity of many-body quantum systems is a formidable and often insurmountable obstacle for the classical computational tools. %developed to simulate the many-electron quantum systems. 
In recent years, algorithms for digital quantum computers that follow Feynman's initial vision \cite{feynman2018simulating,lloyd1996universal} emerged as one of the most promising tools to address the difficulties associated with quantum simulation \cite{bharti2022noisy,childs2018toward,cerezo2021variational}.

\begin{figure*}[h!]
\centering
\begin{tikzpicture}
    \tikzstyle{box} = [rectangle, rounded corners, text centered, text width=2.5cm, draw=black, scale=0.8];
    \tikzstyle{arrow} = [thick, ->, >=stealth];
    \tikzstyle{label} = [anchor=north, text centered, text width=1.3cm, scale=0.8, yshift=-0.1cm]
    \node (system) [box] {Physical system};
    \node (fermih) [box, right of=system, xshift=3.5cm] {Fermionic Hamiltonian};
    \node (qubith) [box, right of=fermih, xshift=3.5cm] {Qubit Hamiltonian};
    \node (clifford) [box, dashed, right of=qubith, xshift=3.5cm] {Optimized Hamiltonian};
    \node (algo) [box, right of=clifford, xshift=3cm] {Simulation, VQE, etc.};
    \draw [arrow] (system) -- node[label] {\footnotesize Second  quantize} (fermih);
    \draw [arrow] (fermih) -- node[label] {\footnotesize JWT, BKT, etc.} (qubith);
    \draw [arrow, dashed] (qubith) -- node[label] {\footnotesize Clifford optimize} (clifford);
    \draw [arrow, dashed] (clifford) -- (algo);
\end{tikzpicture}
\caption{Schematic for representing a physical fermionic system on a quantum device consisting of qubits. The physical system is first modelled using a fermionic Hamiltonian expressed using fermionic second-quantized operators. Conventional fermion-to-qubit mappings allows one to represent the fermionic Hamiltonian using operators acting on qubit degrees of freedom, albeit without taking into account the structure of problem Hamiltonian. Our work focuses on using numerical optimization to design a Clifford circuit,  whose adjoint action on the qubit Hamiltonian minimizes a cost function (dotted lines).}
\label{fig:flowchart}
\end{figure*}
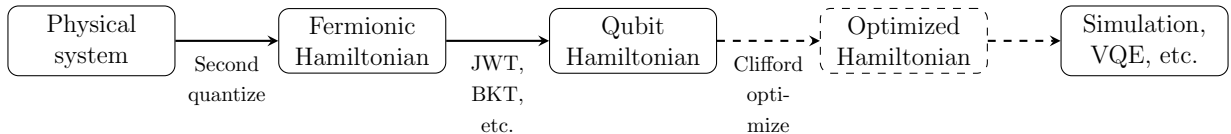

In order to perform digital quantum simulation, one needs to map physical systems onto qubits realized in well-controlled quantum devices. 
Although spin systems can usually be easily mapped to qubits, quantum systems consisting of fermions present a non-trivial mapping problem \cite{bravyi2002fermionic} because the fermionic wave function of electrons is antisymmetric under particle exchange. This fundamental difference between fermions and qubits means that local fermionic operators are generically mapped to non-local qubit operators that have a large Pauli weight (i.e., that act on a large number of qubits).
In turn, non-local qubit operators increase the implementation cost of the fermionic simulation on digital quantum computers, especially when considering compilation of such long Pauli strings into single- and two-qubit gates \cite{paulihedral}. Moreover, even in the context of Variational Quantum Algorithms, where the simulated Hamiltonian is not encoded in the dynamics of the system but rather in the expectation values of observables on a parametrized quantum state \cite{cerezo2021variational,holmes2022connecting}, the nonlocality of the Pauli operators dramatically reduces the effectiveness of the state-of-the-art classical shadow protocols \cite{huang2020predicting}.
 
Tremendous progress has been made in encoding fermions to qubits by considering maps between single fermionic annihilation operators and Pauli strings acting on qubits. These conventional fermion-to-qubit mappings, readily available in widely-used quantum software libraries \cite{mcclean2019openfermion,qiskit_nature}, include the Jordan-Wigner transformation (JWT), the Bravyi-Kitaev transformation (BKT), and the parity map \cite{bravyi2002fermionic, Seeley_2012, Tranter_2018}. The more recent development of ternary-tree mappings\cite{Jiang_2020} 
has resulted in an elegant formalism for understanding the connections between these conventional mappings. Moreover, the ternary-tree formalism revealed that all mappings given by shallow ternary trees are asymptotically optimal with respect to a key performance metric, the Pauli weight averaged over \emph{individual} creation and annihilation operators. Yet crucially, this result is not sufficient for determining which among the shallow ternary-tree mappings results in the lowest average Pauli weight for a specific problem Hamiltonian.

Recent years have witnessed an increased effort to address the challenge of tailoring fermion-to-qubit mappings to a concrete problem Hamiltonian and experimental hardware. Some of these works consider designing mappings that map $n$ fermionic modes to $m\neq n$ qubits, unlike their conventional counterparts. In particular, fermion-to-qubit mappings that leverage global Hamiltonian symmetries to construct mappings on $m<n$ qubits \cite{bravyi2017tapering,fischer2019symmetry}, as well as those that map local fermionic operators to local qubit operators acting on more than $m>n$ qubits on specific hardware lattice geometries were constructed \cite{verstraete2005mapping,whitfield2016local,Steudtner_2018,Setia_2019,derby2021compact,ultrafast2024obrien}.
In a complementary manner, methods for leveraging the structure of the physical chemistry Hamiltonian to obtain simplified problem Hamiltonians were invented \cite{bravyi2017tapering, Babbush_2018}.
Similarly, mappings tailored for a concrete computational task has been developed, including those for Trotterization based simulation \cite{hastings2014improving} maximizing initial gradients for VQE \cite{sun2023towards} or reducing entanglement in ground states \cite{mishmash2023hierarchical}.

The complexity associated with the particularities of a problem Hamiltonian and hardware realization has motivated numerical optimization approaches to designing fermion-to-qubit mapping \cite{Steudtner_2018,chien2022optimizing,nys2022variational,Chiew_2023,miller2024treespilation}. However, these works either focus on a narrow class of problem Hamiltonians \cite{nys2022variational} or a smaller set of fermion to qubit mappings \cite{Steudtner_2018,Chiew_2023,miller2024treespilation}. While the approach presented by Ref. \cite{chien2022optimizing} does not have either one of these restrictions, their use of computationally expensive approximately brute-force optimization over an exponentially large class of mappings drastically reduced the applicability of their approach to translationally invariant systems with small unit cells. Achieving effective numerical optimization of fermion-to-qubit mappings for large systems is the main motivation of this work.

To efficiently leverage numerical optimization methods for designing fermion-to-qubit mappings for large systems, we first translate the problem of designing fermion-to-qubit mappings to the problem of optimizing a unitary transformation of the qubits. The main insight that allows for this is simple: the adjoint action of a unitary on the qubit representation of fermionic operators generates a new fermion-to-qubit mapping. We then focus on a restricted class of unitaries, represented by Clifford circuits, and study the structure of the resulting optimization problem. 
We show that the class of fermion-to-qubit mappings generated by transforming an initial ternary-tree-based mapping \emph{includes} all ternary-tree mappings \cite{Jiang_2020}, and more. In other words, mappings that are parametrized by Clifford circuits include all ternary-tree mappings. We then explore these Cliffor-circuit parametrized mappings using a heuristic simulated annealing algorithm \cite{kirkpatrick1983optimization} to find a circuit that optimizes a given cost function. The optimization step that we focus on this paper is depicted schematically in Fig. \ref{fig:flowchart}.
Optimization of Clifford circuits has two advantages. First, because the Clifford group is by definition the normalizer of the Pauli group \cite{gottesman1998heisenberg}, it maps each Pauli term to a single Pauli term, preventing the growth of the number of Hamiltonian terms during optimization. Second, the classical simulability of Clifford circuits enables an efficient evaluation of the cost function. 

Although our approach is flexible in terms of the choice of the cost function, here we focus on a simple and hardware-independent cost function, namely the average Pauli weight of the resulting qubit Hamiltonian, for all examples we consider. The heuristic algorithms do not come with the promise of finding a globally optimal mapping. Yet, we demonstrate that our approach allows one to efficiently explore the space of high-performing fermion-to-qubit mappings for a broad range of problem Hamiltonians with more than 1,500 terms, without any additional tapering \cite{bravyi2017tapering,fischer2019symmetry} of the Hamiltonian prior to the optimization step.

The numerical optimization of fermion-to-qubit mappings parametrized by Clifford circuits
results in improvements that are both of practical and theoretical importance. To quantify these improvements, we compare the average Pauli weight associated with the fermionic problem Hamiltonian under the optimized mapping to that under three conventional mappings: JWT, BKT, and shallow ternary-tree mappings. Specifically, we report the improvement due to the optimized mapping using the lowest percent reduction of the Pauli weight. All results were obtained through simulated annealing calculations with each optimization run taking at most 3 days on a single CPU.
The improvements that we obtained over a broad range of problem Hamiltonians highlights the potential of the heuristic optimization approach for designing efficient customized fermion-to-qubit mappings.

On the practical side, we demonstrate improvements on one- and two-dimensional hopping models with open boundary conditions.
In the case of one-dimensional hopping Hamiltonians with various hopping ranges, our optimization approach improves the constant overhead of shallow ternary-tree mappings. Specifically, our optimized mappings result in a percent reduction of $5\%$ to $10\%$ for all-to-all coupled systems of up to 20 sites. Interestingly, the relatively modest performance of the optimal mappings for all-to-all coupled hopping models is substantially improved as we reduce the range of hopping. In particular, for systems of size 10--20, we find that the performance of the optimized mapping is peaked around the hopping range $r=6$. This result highlights that our optimized mappings manage to leverage the structure of the problem of intermediate complexity. Turning our attention to two-dimensional hopping models, we find even more remarkable improvements. We find that for $L\times L$ nearest-neighbor hopping models, we obtain a percent reduction of up to $45\%$ for $L=6$ (120 terms) and up to $35\%$ for $L=8$ (224 terms).  
Importantly, we find that the percent reduction in two-dimensional hopping models is not drastically reduced with the introduction of on-site interactions. We find that for two-dimensional interacting Hubbard models of up to 36 sites (349 Hamiltonian terms), the percent reduction is 25\%, compared to the ternary-tree mapping. Finally, 
we provide numerical evidence that the optimized mappings result in a sizable percent reduction for intermediate-scale chemistry Hamiltonians, in particular one-dimensional Hydrogen chains. For a 6-site Hydrogen chain, described by approximately 1500 terms, we find a percent reduction of $\approx 10-20\%$, compared to conventional mappings. 

A more detailed study of the Hydrogen chain led us to a discovery of theoretical importance. We observe that none of the fermion-to-qubit mappings optimized for hydrogen chains satisfied a necessary condition to be classified as a ternary-tree mapping \cite{Jiang_2020}. We were thus motivated to answer the following question: ``Is there a simple model for which our optimized mappings overperform \textit{all} ternary-tree mappings?". We found that the answer to this question is affirmative. Through an exhaustive search over all ternary-tree mappings, we verified that for a system consisting of 4 fermionic modes interacting through a simple four-body interaction term, the optimized fermion-to-qubit mapping outperforms \emph{all} ternary-tree mappings by at least $20\%$.

The results of this paper are organized as follows. In Sec. \ref{sec:methods}, we present details of our fermion-qubit mapping methods and the proposed optimization algorithm. Sec. \ref{sec:results} presents key analytical results with a comparison to ternary-tree mapping, followed by numerical results demonstrating the performance of our algorithm for one- and two-dimensional hopping and interacting fermionic lattice models and hydrogen chains. We conclude with an outlook for utilizing numerical optimization for designing customized fermion-to-qubit mappings in Sec. \ref{sec:conclusion}.

%%%%%%%%%%%%%%%%%%%%%%%%%%%%%%%%%%%%%%%%%%%
\section{Methods}
\label{sec:methods}

We begin in Sec. \ref{sec:notation} by setting up a formal mathematical framework for fermion-to-qubit mappings. In Sec. \ref{sec:unitary}, we use this framework to relate different mappings to each other via unitary transformations, and introduce a corresponding Hamiltonian-dependent optimization problem over unitaries. In Sec. \ref{sec:clifford}, we restrict the search space to the more manageable and discrete Clifford group and briefly discuss the numerical optimization algorithm used to produce the numerical results in Section $\ref{sec:results}$.

\subsection{Notation and background}
\label{sec:notation}

Fermion-to-qubit mappings allow one to use quantum processors consisting of local qubit degrees of freedom to address problems which are naturally described with non-local fermionic degrees of freedom \cite{bravyi2002fermionic, Tranter_2018}.

We take $\mathcal{H}_f$ as the space of second-quantized Fock states on $n_{\rm f}$ ordered fermionic modes $\ket{\lambda}$. This is a $2^{n_f}$-dimensional Hilbert space that can be decomposed as
$
\mathcal{H}_f = \bigoplus_{m=1}^{n_f} \mathcal{H}^{m}_f,$
where $\mathcal{H}^{m}_f$
is an $m-$fermion Hilbert space spanned by $\Lambda_{\rm as}(\ket{\lambda_1, \ldots, \lambda_m})$ where $\Lambda_{\rm as}$ is the antisymmetrizes the state with respect to particle exchange \cite{bravyi2002fermionic,altland2010condensed}.
In second-quantized form, these basis states are written in terms of Fock states $\ket{f_1, \ldots, f_{n_f}}$ where $f_k \in \{0, 1\}$ are occupation numbers. Given the system has exactly $m$ fermions and $\{j_1, \ldots, j_m\}$ is the set of indices $j$ for which $f_j = 1$, then 
\begin{equation}
\ket{f_1, \ldots, f_{n_f}} = \Lambda_{\rm as} \left( \ket{\lambda_{j_1}, \ldots, \lambda_{j_m}} \right).
\end{equation}

Fermionic operators acting on $\mathcal{H}_f$ can be conveniently expressed in terms of creation and annihilation operators, $a_j^\dagger$ and $a_j$, respectively. These obey the canonical anti-commutation relations
\begin{align}
    \{a_i,a_j\} &=  \{a_i^{\dagger},a_j^{\dagger}\} = 0\\
     \{a^{\dagger}_i,a_j\} &= \delta_{ij} \mathbf{I},
\end{align}
where $\delta_{ij}$ is the Kronecker delta and $\mathbf{I}$ is the identity operator. 
In contrast, the qubit space $\mathcal{H}_q$ is a $2^{n_q}$-dimensional Hilbert space spanned by computational basis states $\ket{q_1, \ldots, q_n}$, where each $q_j \in \{0, 1\}$ with $j= [n_{\rm q}]$.

Given the fermionic and qubit Hilbert spaces $\mathcal{H}_f$ and $\mathcal{H}_{q}$, a \emph{fermion-to-qubit map} is an isometry (i.e. linear map preserving inner products) $\phi \colon \mathcal{H}_f \to \mathcal{H}_q$.
%for some Hilbert space $\mathcal{H}_q$ on $n$ qubits.
In the following, we use $\im$ to denote the image of a map and $\mathcal{L}(\mathcal{H})$ to denote the space of operators acting on $\mathcal{H}$. 
Given a fermionic operator $\mathcal{O}_f \in \mathcal{L}(\mathcal{H}_f)$, the corresponding qubit operator $\mathcal{O}_q \in \mathcal{L}(\im \phi)$ under this map is defined as 
\begin{equation}
\mathcal{O}_q = \phi \circ \mathcal{O}_f \circ \phi^{-1}
\label{eq:fermion-op-map}
\end{equation}
so that it acts on qubit states in the same manner as $\mathcal{O}_f$ acts on fermionic states:
\begin{equation}
\phi(\mathcal{O}_f\ket{\psi}_f) = \mathcal{O}_q \phi(\ket{\psi}_f).
\end{equation}
We denote by $\Phi \colon \mathcal{L}(\mathcal{H}_f) \to \mathcal{L}(\im \phi)$ the induced operator mapping $\mathcal{O}_f \mapsto \mathcal{O}_q$ defined in Eq. (\ref{eq:fermion-op-map}).

In this paper we focus on the case where $n_{\rm f}=n_{\rm q}$, in which case $\phi$ is a surjective isometry (i.e., a unitary transformation). Nevertheless, we remark that our methods fully generalize to cases where $n_{\rm f}\neq n_{\rm q}$ \cite{Steudtner_2018,Setia_2019,nys2022variational,derby2021compact,fischer2019symmetry}. In some applications, the relevant space of interest is a strict subspace of the fermionic Hilbert space of $n_{\rm f}$ modes ( e.g. if the Hamiltonian of interest has certain symmetries which restrict the permissible Hilbert space ) \cite{Steudtner_2018,fischer2019symmetry}
%\todo{@Yuan: Could you put the relevant refs?}
. In such cases, one can obtain maps with $n_{\rm f}> n_{\rm q}$. On the other hand, there are also benefits to consider maps with $n_{\rm f}< n_{\rm q}$, such as introducing error correction \cite{Setia_2019,nys2022variational,derby2021compact}. %\todo{@Yuan: Could you put the relevant refs?}. 
Moving forward, we take $n_f = n_q = n$.

The most elementary fermion-to-qubit mapping is the Jordan Wigner transformation (JWT), which maps fermionic states to qubit states by defining $q_j = f_j$ for each $1 \le j \le n$. The induced operator mapping $\Phi$ maps individual creation and annihilation operators to  strings of Pauli operators
\begin{align}
\Phi_{\rm JW}(a_j) = \frac{1}{2} Z_1 \cdots Z_{j-1}(X_j - i Y_j).
\end{align}

Crucially, because of the non-trivial exchange statistics associated with the fermionic degrees of freedom, their representation in terms of their qubit counterparts involve non-local operations. The non-locality of the resulting qubit operators poses a challenge for implementing the quantum algorithms that involve chemistry Hamiltonians \textcolor{red}{\cite{tranter2015b}}. 
Conventionally, the non-locality of the qubit operators can be quantified by first expressing the qubit operator as a sum of Pauli strings and then calculating the average Pauli weight.
Here, a \emph{Pauli string} is a tensor product operator over $n_{\rm q}$ qubits where each factor is a Pauli $X$, $Y$, $Z$, or $I$. The Pauli strings acting on $n_{\rm q}$ qubits form a complete basis for all operators that act on the system. The \emph{Pauli weight} of a single Pauli string $P$, which we denote by $\wt(P)$, is the number of factors that are not $I$. 
When considering the Pauli weight associated with an operator $\mathcal{O}_q$, which is a linear combination of multiple Pauli strings $\{P_j\}$, 
\begin{align}
\mathcal{O}_q = \sum_{j=1}^{\ell} c_j P_j, \quad \quad c_j \in \mathbf{C}/\{0\},
\end{align}
then the \emph{average Pauli weight} is given by
\begin{equation}
\wt(\mathcal{O}_q) = \frac{1}{\ell} \sum_{j=1}^\ell \wt(P_j).
\end{equation}

It is possible to evaluate the efficiency of different fermion-to-qubit mappings based on a comparison of the average or worst-case weights for $\wt(\Phi(a_j))$ across all $1 \le j \le n$. For the JWT, both metrics are $O(n)$.
This is reduced dramatically if we consider mappings such as the BKT \cite{Tranter_2018} and its ternary-tree-based generalizations \cite{Jiang_2020}. These mappings are all near-optimal in that the Pauli weight averaged over single creation and annihilation operators grows as $O(\log n)$. In addition, Jiang et al. \cite{Jiang_2020} found ternary-tree-based mappings with an average Pauli weight of $\log_3{(2n)}$, and proved that it is not possible to construct a fermion-to-qubit mapping whose average Pauli weight over single creation and annihilation operators is smaller than the optimal ternary-tree mapping. Determining whether the ternary-tree-based mappings are optimal in more general contexts is one of the main goals of our work.

\subsection{Design of fermion-to-qubit mappings through unitary transformations}
\label{sec:unitary}

In this work, we are interested in qubit operators that correspond to electronic structure Hamiltonians, which have the form \cite{Tranter_2018}
\begin{equation}
H_f = \sum^N_{i,j} h_{ij} a_i^\dagger a_j + \sum_{i,j,k,l}^N h_{ijkl} a_i^\dagger a_j^\dagger a_k a_l,
\label{eq:elec-struc-ham}
\end{equation}
where $N$ is the number of single-particle orbitals. Note that the physical Hamiltonian $H_f$ contains only operators that conserve the fermion number \cite{bravyi2002fermionic}.
We refer to the $a_i^\dagger a_j$ terms as one-body \textit{hopping terms} and the $a_i^\dagger a_j^\dagger a_ka_l$ as two-body \textit{interaction terms}.

The goal is to find an efficient fermion-to-qubit mapping $\phi$ which minimizes the resulting Pauli weight of $H_f$ under the induced operator mapping $\Phi$:
\begin{equation}
\argmin_{\phi \colon \mathcal{H}_f \to \mathcal{H}_q} \wt(\Phi(H_f)).
\end{equation}
Thus, our approach for designing fermion to qubit mappings take is to numerically search the space of $\phi$ to minimize the cost function wt given a fermionic Hamiltonian $H$.

The space of valid fermion-to-qubit mappings is fully specified through unitary transformations acting on an initial fermion-to-qubit mapping. This fact is summarized in the following lemma.

\begin{lemma}\label{thm:unitary-transform}
Let $\phi_1, \phi_2 \colon \mathcal{H}_f \to \mathcal{H}_q$ be two fermion-to-qubit maps on a given fermion and qubit space. Then there exists a unitary $U \in \mathcal{L}(\mathcal{H}_q)$ such that $\phi_2 = U \phi_1$. The corresponding operator maps $\Phi_1$ and $\Phi_2$ transform as 
\begin{equation}
\Phi_2(\mathcal{O}_f) = U \Phi_1(\mathcal{O}_f) U^\dagger.
\label{eq:unitary-transform}
\end{equation}
\end{lemma}
Moreover, given a fermion-to-qubit map $\phi$, and a unitary $U$, $U\phi$ is a valid fermion-to-qubit map.

\begin{proof}
Consider the operator $U = \phi_2 \phi_1^{-1}$. Since $\phi_1$ and $\phi_2$ are isometries, $U$ is a unitary transformation from $\im(\phi_1)$ to $\im(\phi_2)$. If the $\phi$ were not surjective, then extend $\im(\phi_1)$ and $\im(\phi_2)$ each into an orthonormal basis for $\mathcal{H}_q$. Consequently, mapping the extended bases to each other extends $U$ to a unitary on $\mathcal{H}_q$. For the operator maps, we have
\begin{equation}
\begin{aligned}
\Phi_2(\mathcal{O}_f) &= \phi_2 \mathcal{O}_f \phi_2^{-1} \\ &= U \phi_1 \mathcal{O}_f \phi_1^{-1} U^\dagger \\
&= U \Phi_1(\mathcal{O}_f) U^\dagger. \qedhere
\end{aligned}
\end{equation}
\end{proof}

Therefore, instead of designing fermion-to-qubit maps from scratch, we may take any initial mapping of operators $\Phi_1$ and apply unitary transformations. 
From this perspective, it is straightforward to formulate the following optimization problem to find a fermion-to-qubit mapping that minimizes wt: 
\begin{problem}[Optimal fermion-to-qubit mapping]
    Given an initial $n$-qubit Hamiltonian $H_q$ find a unitary transformation $U$ that optimizes the cost function
\begin{align}
    \argmin_{U \in \mathrm{U}(2^{n})} \mathrm{wt}(UH_qU^\dagger).
    \label{eq:Minimization General}
\end{align}
Here, $H_q = \Phi_0(H_f)$ is a qubit Hamiltonian and $\Phi_0$ is any existing fermion-to-qubit map, such as JWT or BKT.
\end{problem} 

Although we focus solely on the Pauli weight of the resulting Hamiltonian $H_q$ as our cost function, we emphasize that the optimization approach allows for searching an effective fermion-to-qubit mapping given a broader range of constraints on a possible implementation scheme, such as connectivity of and range of interactions between qubits \cite{cross2019validating}. 
Furthermore, our approach is not limited to Hamiltonians; we can apply fermionic maps and unitary transformations on any fermionic operator $\mathcal{O}_f$, even non-Hermitian operators. For example, the fermionic operator
\begin{equation}
\mathcal{O}_f = \sum_{i=1}^n a_i,
\label{eq:single-ops}
\end{equation}
exactly reproduces the standard notion of average Pauli weight across individual fermionic operators.

%%%%%%%%%%%%%%%

\subsection{Clifford transformations}
\label{sec:clifford}

The expressivity of our scheme comes with many challenges familiar to variational quantum algorithms \cite{holmes2022connecting,cerezo2021variational}. In particular, the exponential number of parameters involved in describing a generic $n_q$-qubit unitary will result in the emergence of barren plateaus as well as a plethora of local minima in the landscape of the cost function. Moreover, the evaluation of the cost function itself is computationally expensive for a generic unitary transformation. To address these challenges, in the next subsection, we constrain the space of fermion-to-qubit transformations to subsets of Clifford transformations, whose action on the qubit Hamiltonian can be efficiently simulated.

First, we restrict our search to the space to unitaries $U_{\rm Clifford}$ in the Clifford group, which can be expressed as 
\begin{align}
    U_{\rm Clifford}=G_M\cdots G_1 G_0 \equiv \Tilde{\prod}^M_{n=1} G_n.
\label{eq:clifford-unitary}
\end{align}
Here, $\tilde{\Pi}_n$ denotes the ordered product of the Clifford gates $G_n$ which are chosen from the set
\begin{align}
    \mathcal{S}_{\rm Clifford} = \{{\rm H}_i, {\rm S}_i,{\rm CNOT}_{ij} \mid 1 \le i \ne j \le n\},
\end{align}
consisting of the Hadamard gate H, phase gate S, and the controlled-NOT gate CNOT$_{ij}$, with the indices $i^{\rm th}$ and $j^{\rm th}$ referring to the control and target qubits, respectively. 

The Clifford group allows us to leverage its well-known properties to expedite the numerical optimization of Eq. ($\ref{eq:Minimization General}$). Most importantly, the action of these transformations on Pauli strings can be  efficiently simulated by a classical computer, via the Gottesman-Knill theorem \cite{gottesman1998heisenberg,aaronson2004improved}, circumventing the need to perform exponentially large matrix multiplications. In addition, a Clifford circuit maps a single Pauli string to a single Pauli string, preserving the total number of terms. Indeed, the converse also holds. That is, the Clifford group is the unitary normalizer of the $n$-qubit Pauli group \cite{nielsen2010quantum}.
Lastly, the Clifford group has a simpler representation of the search space as it is generated by a small discrete set of gates.

We can further restrict the set of Clifford generators without a drastic effect on the performance of the optimal fermion-to-qubit mapping. To this end, observe that 
while the adjoint action of CNOT gates can reduce the weight of a Pauli string (see Fig. \ref{fig:cnot}), the single-qubit Hadamard and phase gates do not. Motivated by this observation, we consider three potential gate sets and their combinations. First, we consider 
\begin{align}
    \mathcal{S}_{C} &= \{{\rm CNOT}_{ij} \mid 1 \le i \ne j \le n\} \label{eq: SC}.
\end{align}
Then, as a way of extending the search space, we consider the union of $\mathcal{S}_C$ with the following sets of generators
\begin{align}
\mathcal{S}_{H} &= \{{\rm CNOT}_{ij} {\rm H}_i \mid 1 \le i \ne j \le n\}\\
\mathcal{S}_{S} &= \{{\rm CNOT}_{ij} {\rm S}_i \mid 1 \le i \ne j \le n\}
\end{align}
In the following, we consider the following unions of these restricted gate sets:
\begin{align}
    \mathcal{S}_{CH} &=\mathcal{S}_{C}\cup\mathcal{S}_{H} \\ 
    \mathcal{S}_{CHS} &=\mathcal{S}_{C}\cup\mathcal{S}_{H} \cup\mathcal{S}_{S}. \label{eq: SCHS}
\end{align}
We further discuss the structure of gate sets $\mathcal{S}_{C}$, $\mathcal{S}_{CH}$, and $\mathcal{S}_{CHS}$, and provide a performance evaluation in  Sec.~\ref{sec:single-ops} and Sec.~\ref{sec:comparison}.

We defer to Appendix \ref{sec:algorithms} for a detailed discussion of the simulated annealing algorithm used to optimize the cost function in Eq. ($\ref{eq:Minimization General}$). In brief, our optimization algorithm builds a Clifford unitary $U_{\rm Clifford}$ iteratively by adding gates from a specified subset of $\mathcal{S}_{\rm Clifford}$ to an initial identity operation. Our numerical simulations demonstrate that $\mathcal{S}_{CHS}$ does not provide a substantial improvement compared to $\mathcal{S}_{CH}$ when it comes to reducing the average Pauli weight, although $\mathcal{S}_{CHS}$ can generate all Clifford transformations while $\mathcal{S}_{CH}$ cannot.

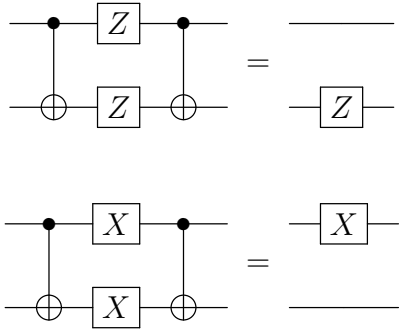
\begin{figure}[t]
\[
\Qcircuit @R=0.7em @C=1.0em {
& \ctrl{2}	& \gate{Z}	& \ctrl{2} 	& \qw 	& 	& & \qw & \qw \\
& 			&			&			&		& = & &	&	&	\\	
& \targ 	& \gate{Z}	& \targ 	& \qw 	& 	& & \gate{Z} & \qw
}
\]

\[
\Qcircuit @R=0.7em @C=1.0em {
& \ctrl{2}	& \gate{X}	& \ctrl{2} 	& \qw 	& 	& & \gate{X} & \qw \\
& 			&			&			&		& = & &	&	&	\\	
& \targ 	& \gate{X}	& \targ 	& \qw 	& 	& & \qw & \qw
}
\]
\caption{The conjugation action of CNOT reduces the weight of some Pauli strings}
\label{fig:cnot}
\end{figure}

%%%%%%%%%%%%%%%%%%%%%%%%%%%%%%%%%%%%%%%%%%%
\section{Results and Discussion}
\label{sec:results}

In this section, we present a detailed characterization of the space of mappings that is reachable through the application of Clifford unitaries on a Jordan-Wigner mapping, as well as numerical evidence demonstrating how our optimization approach results in fermion-to-qubit mappings that overperform the state-of-the-art mappings in the context of condensed-matter and chemistry Hamiltonians. In Sec.~\ref{sec:single-ops}, we analytically show that Clifford gates can map any initial ternary-tree mapping to a desired ternary-tree mapping. 
Moreover, the Clifford gates also allow one to explore non-ternary-tree mappings, a property that proves to be very useful when considering fermionic Hamiltonians with exchange interactions.

In the remaining three subsections, we present numerical results for hopping models, Hubbard models, and hydrogen chains. Each optimized mapping was obtained in at most 3 days of simulated annealing computation on a single Intel Xeon Silver 4216 CPU. Though each data point is the best of many runs, each run is independent of other runs, and thus can be performed in parallel on a cluster.

\subsection{Comparison to ternary-tree mapping}
\label{sec:single-ops}

In this section, we consider the optimization problem under the fermionic operator $\mathcal{O}_f$ in Eq. (\ref{eq:single-ops}), for which the ternary-tree is known to be optimal\cite{Jiang_2020}. We are able to recover its optimal average Pauli weight by applying our simulated annealing algorithm with $\mathcal{S}_C$ (only CNOT gates) to $\mathcal{O}_f$, starting from either JWT or BKT. This result demonstrates that, for $\mathcal{O}_f$, our optimization algorithm finds the global minimum despite being restricted to Clifford transformations and a lack of provable guarantees.

The success of our approach can be explained through a theoretical analysis of our algorithm's search space with respect to the broad class of mappings described by ternary trees, which includes both the JWT and BKT, as well as the optimal balanced ternary tree. We show that our algorithm using $\mathcal{S}_{CHS}$ can generate this family of trees in the search space, while also including mappings outside of this family. Furthermore, we show that the search space using $\mathcal{S}_C$ can generate any tree shape, up to permutations of the qubit and leaf labels.

To this end, we will use the standard language of trees and graph theory, which we briefly review. We study \emph{full} ternary trees, meaning each node has either $0$ or $3$ children. A node with $0$ children is a \emph{leaf}, and a node with $3$ children is a \emph{parent}. If there are $n$ parents, then there are $2n+1$ leaves. A fermion-to-qubit mapping arises\cite{Jiang_2020, bonsai2024miller} from a ternaryvtree with the following information:
\begin{enumerate}
\item A full ternaryvtree, which we call the \emph{shape} of the ternary-tree mapping.
\item A labeling of the $n$ parents by $1, \ldots, n$. Each labeled parent corresponds to a qubit.
\item A labeling of the $2n+1$ leaves by $\gamma_1, \ldots, \gamma_{2n+1}$. These label the Majorana operators.
\end{enumerate}
Each leaf $\gamma_j$ is associated with a Pauli string $P_{\gamma_j}$ by traversing down the tree from the root to $\gamma_j$. Along the path down to $\gamma_j$, taking the left, middle, or right child of qubit $k$ appends to $P_{\gamma_j}$ a Pauli $X_k$, $Y_k$, or $Z_k$, respectively. The mapping is then defined by
\begin{equation}
\Phi(a_k) = P_{\gamma_{2k-1}} + i P_{\gamma_{2k}}
\end{equation}
for each $1 \le k \le n$. The last leaf, $\gamma_{2n+1}$, is redundant. An important property of ternary-tree mappings is that
each pair of distinct Pauli strings $P_{\gamma_j}, P_{\gamma_k}$ anticommute on exactly one qubit. 

At this point, we make a distinction between the general class of mappings described by ternary trees, and the specific ternary-tree mapping that optimizes the average Pauli weight of Eq. (\ref{eq:single-ops}). What we have described so far is a general ternary-tree mapping, which can express many common mappings.\cite{Vlasov_2022, bonsai2024miller} Indeed, the Jordan-Wigner tree is described by a spine of right children, and the Bravyi-Kitaev tree is formed from the first $n$ nodes from a preorder traversal of a perfect binary tree with height $\floor{\log_2 n}$. 
In both cases, the qubits are labeled by an inorder traversal of the parents, and the leaves are labeled by an inorder traversal of the leaves, as shown in Figures \ref{fig:jw-tree} and \ref{fig:bk-tree}. We refer to any standard algorithms text\cite{clrs} for an exposition on tree traversals.
\begin{figure}[t]
\centering
\begin{tikzpicture}
[style = {level distance = 1cm}]
\node (root) [draw, circle] {1}
    child {node {$\gamma_1$}} child {node {$\gamma_2$}} child {node [draw, circle] {2}
    child {node {$\gamma_3$}} child {node {$\gamma_4$}} child {node [draw, circle] {3}
    child {node {$\gamma_5$}} child {node {$\gamma_6$}} child {node {$\gamma_7$}}
    }};
\end{tikzpicture}
\caption{Jordan-Wigner tree on $3$ qubits}
\label{fig:jw-tree}
\end{figure}
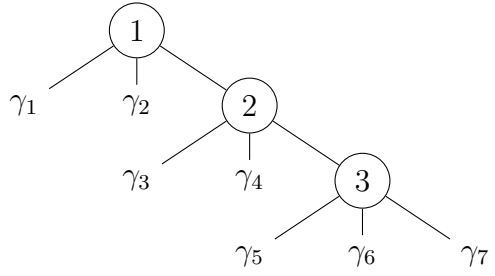
\begin{figure}[t]
\centering
\begin{tikzpicture}
[style = {level distance = 1cm}, level 2/.style = {sibling distance = 1cm}, level 3/.style = {sibling distance = 0.7cm}]
\node (root) [draw, circle] {4}
    child {node [draw, circle] {2}
        child {node [draw, circle] {1}
            child {node {$\gamma_1$}} child {node {$\gamma_2$}} child {node {$\gamma_3$}}
        } 
        child {node {$\gamma_4$}}
        child {node [draw, circle] {3}
            child {node {$\gamma_5$}} child {node {$\gamma_6$}} child {node {$\gamma_7$}}
        }
    }
    child {node {$\gamma_8$}}
    child {node [draw, circle] {5}
        child {node {$\gamma_9$}} child {node {$\gamma_{10}$}} child {node {$\gamma_{11}$}}
    };
\end{tikzpicture}
\caption{Bravyi-Kitaev tree on $5$ qubits}
\label{fig:bk-tree}
\end{figure}
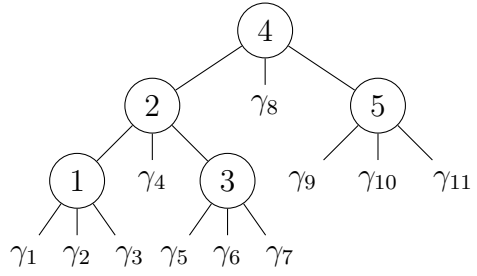

In contrast, the optimal ternary tree described described by Jiang et al.\cite{Jiang_2020} is a balanced ternary tree, where each layer from top to bottom is completely filled with parent nodes before moving on to the next. Though not specified by Jiang et al., for our purposes we refer to \emph{the optimal ternary-tree mapping} as the mapping given by the balanced ternary tree with nodes flushed right, qubits labeled from top to bottom then left to right, and leaves labeled in an inorder traversal, as shown in Figure \ref{fig:ternary-tree}.
We emphasize that neither the position of the parents in the bottom layer nor the labeling of the nodes and leaves affect the Pauli weight of the mapping associated with the tree, averaged over single creation and annihilation operators.

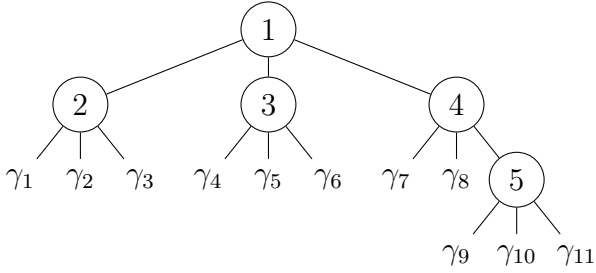
\begin{figure}[t]
\centering
\begin{tikzpicture}
[style = {level distance = 1cm}, level 1/.style = {sibling distance = 2.5cm}, level 2/.style = {sibling distance = 0.8cm}]
\node (root) [draw, circle] {1}
    child {node [draw, circle] {2}
        child {node {$\gamma_1$}} child {node {$\gamma_2$}} child {node {$\gamma_3$}}
    } 
    child {node [draw, circle] {3}
        child {node {$\gamma_4$}} child {node {$\gamma_5$}} child {node {$\gamma_6$}}
    }
    child {node [draw, circle] {4}
        child {node {$\gamma_7$}} child {node {$\gamma_{8}$}} child {node [draw, circle] {5}
            child {node {$\gamma_9$}} child {node {$\gamma_{10}$}} child {node {$\gamma_{11}$}}
        }
    };
\end{tikzpicture}
\caption{An optimal ternary-tree on $5$ qubits}
\label{fig:ternary-tree}
\end{figure}

We are interested in how qubit operators transform under Clifford conjugation. Since each Pauli string in the Hamiltonian corresponds to a product of Majorana operators, the transformation can be understood by studying how the Paulis $\{P_{\gamma_j}\}$ transform under Cliffords. Concretely, given a Clifford $U$, we study the map $\gamma_k \mapsto UP_{\gamma_k}U^\dagger$.
In some cases, the resulting Paulis $UP_{\gamma_k}U^\dagger$ are also given by a tree. In the following, we will call this the transformation of a ternary-tree mapping by the Clifford $U$.
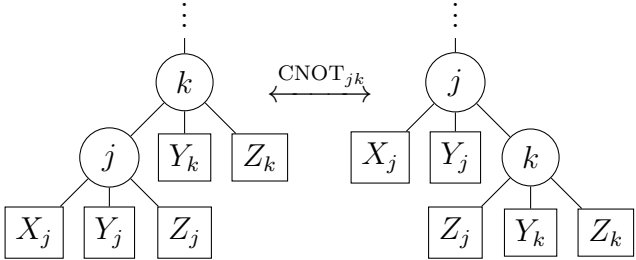
\begin{figure}[h]
\centering
\begin{tikzpicture}
[style = {level distance = 1cm}, level 2/.style = {sibling distance = 1cm}]
\node (before) [xshift=-1.8cm] {$\vdots$} child {
node [draw, circle] {$k$}
    child {node [draw, circle] {$j$}
        child {node [draw] {$X_j$}}
        child {node [draw] {$Y_j$}}
        child {node [draw] {$Z_j$}}
    } 
    child {node [draw] {$Y_k$}}
    child {node [draw] {$Z_k$}}
};
\node [yshift=-1cm] {$\xleftrightarrow{\CNOT_{jk}}$};
\node [xshift=1.8cm] {$\vdots$} child {
node [draw, circle] {$j$}
    child {node [draw] {$X_j$}}
    child {node [draw] {$Y_j$}}
    child {node [draw, circle] {$k$}
        child {node [draw] {$Z_j$}}
        child {node [draw] {$Y_k$}}
        child {node [draw] {$Z_k$}}
    } 
};
\end{tikzpicture}
\caption{Tree rotation corresponding to the adjoint action of CNOT$_{jk}$ transformation. The square boxes can be leaves or contain further subtrees}
\label{fig:tree-rotation}
\end{figure}

\begin{figure}[h]
\centering
\begin{tikzpicture}
[style = {level distance = 1cm}, level 2/.style = {sibling distance = 1cm}]
\node (before) [xshift=-1.8cm] {$\vdots$} child {
node [draw, circle] {$k$}
    child {node [draw] {$X_k$}}
    child {node [draw, circle] {$j$}
        child {node [draw] {$X_j$}}
        child {node [draw] {$Y_j$}}
        child {node [draw] {$Z_j$}}
    } 
    child {node [draw] {$Z_k$}}
};
\node [yshift=-1cm] {$\xleftrightarrow{\CNOT_{jk}}$};
\node [xshift=1.8cm] {$\vdots$} child {
node [draw, circle] {$k$}
    child {node [draw] {$X_k$}}
    child {node [draw] {$Z_j$}}
    child {node [draw, circle] {$j$}
        child {node [draw] {$Y_j$}}
        child {node [draw] {$X_j$}}
        child {node [draw] {$Z_k$}}
    } 
};
\end{tikzpicture}
\caption{A CNOT gate on a middle child moves it to a right child and rotates some other children. The square boxes can be leaves or contain further subtrees}
\label{fig:tree-middle}
\end{figure}
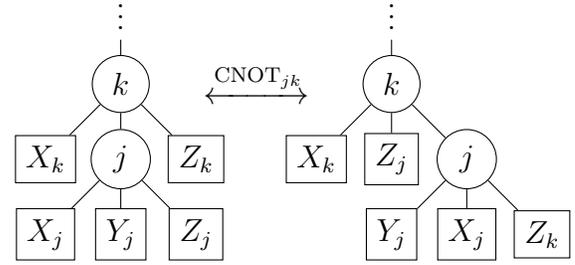

The key transformation is that induced by a $\CNOT$, which rotates adjacent nodes, as shown in Figures \ref{fig:tree-rotation} and \ref{fig:tree-middle}. Using these, we can construct a sequence of CNOT gates to transform between any two ternary-tree shapes, as well as the exact JWT and BKT trees.

\begin{theorem}
\label{thm:cnot-tree-shape}
Let $T_1$ and $T_2$ be any two ternary-tree mappings. There exists a sequence of $\CNOT$ gates, $(\CNOT_{c_i, t_i})_{i=1}^{i_{\max}}$, such that the Clifford $U = \prod\limits_{i=1}^{i_{\max}} \CNOT_{c_i, t_i}$ transforms $T_1$ into the tree shape of $T_2$.
\end{theorem}

As mentioned before, the average Pauli weight of \textit{single} fermionic operators, as defined by Eq. (\ref{eq:single-ops}), is determined only by the shape of the tree and not by the labeling of the qubits or Majoranas. Thus, Theorem \ref{thm:cnot-tree-shape} explains why our search consisting of only CNOT gates is always able to attain the optimal single-operator average weight starting from any ternary-tree mapping, including JWT or BKT. As a corollary, we can implement an exact transformation between ternary-tree mappings that share the same inorder transversal of leaves and qubits, such as the JWT and BKT, using only $\CNOT$ gates because the $\CNOT$ tree rotations preserve the inorder traversal of the qubits and leaves.

\begin{theorem}
\label{thm:jw-bk}
There exists a sequence of $\CNOT$ gates, $(\CNOT_{c_i, t_i})_{i=1}^{i_{\max}}$, such that the Clifford $U = \prod\limits_{i=1}^{i_{\max}} \CNOT_{c_i, t_i}$ transforms between the Jordan-Wigner and Bravyi-Kitaev trees on $n$ qubits.
\end{theorem}

However, in more realistic Hamiltonians, the Majorana labelings do matter, as only products of certain pairs or quadruples of Majoranas are included. The following result states that the broader Clifford group can fix any desired labeling.

\begin{theorem}
\label{thm:trees-included}
For any two ternary-tree mappings $T_1$ and $T_2$, there exists a Clifford circuit which transforms between them.
\end{theorem}
Theorem \ref{thm:trees-included} implies that the space of all possible ternary-tree mappings is included in our algorithm's search space when using the full set of Clifford gates.

Finally, we emphasize that the transformation of a ternary tree under a Clifford circuit does not always result in another ternary tree. Hence, our search goes beyond the space of ternary trees. For example, a CNOT gate between two non-adjacent nodes in a ternary-tree mapping will in general break the tree structure by producing Majoranas which anticommute on more than one qubit. For some Hamiltonians, our algorithm will return to a ternary-tree mapping even after initially traversing outside the space during the search. Moreover, as we detail in Section \ref{sec:hydrogen}, the fact that our search extends beyond the ternary-tree mappings proves useful for certain classes of chemistry Hamiltonians that include exchange interactions. In particular, for such Hamiltonians, we obtain a mapping that cannot be expressed as a ternary tree, but which outperforms every possible ternary-tree mapping. In the following three sections, we detail the application of our optimal fermion-to-qubit mapping approach on electronic structure Hamiltonians that model lattice systems and hydrogen chains.

%%%%%%%%%%%%%%%%%%%%%%%%%%%%%%%%%

\begin{figure*}[h!]
\includegraphics[width=\textwidth]{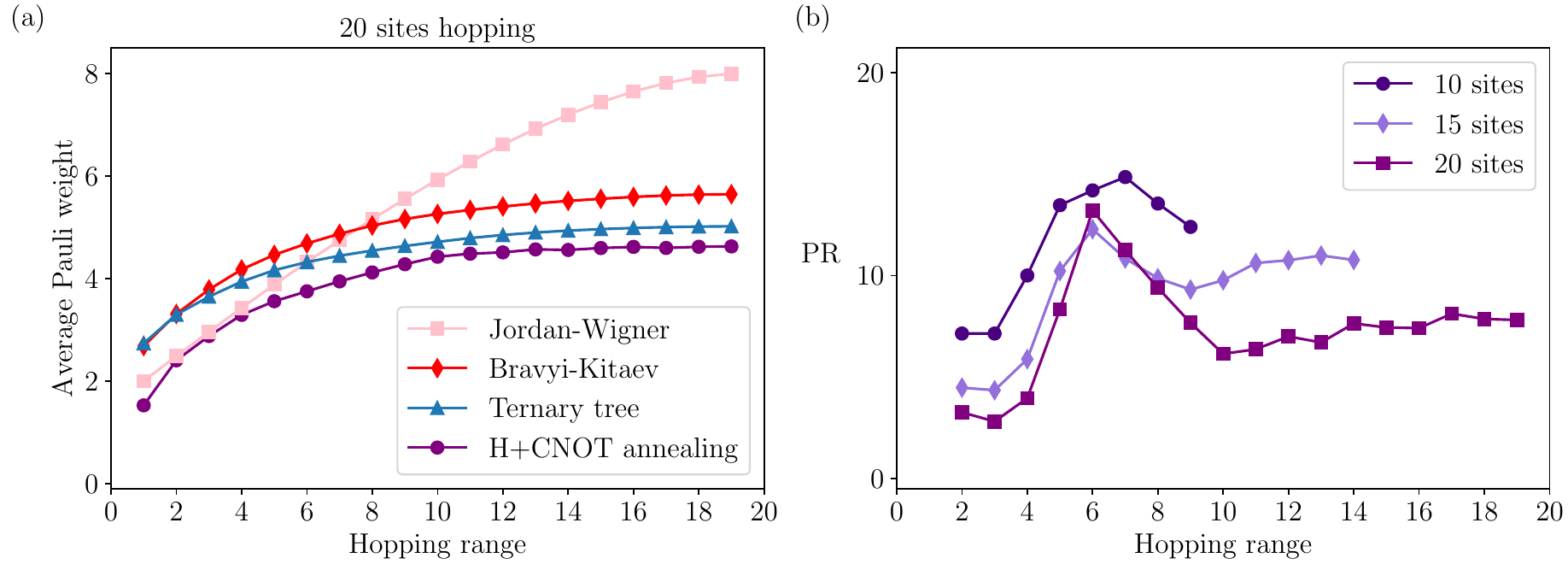}
\caption{For the hopping model in Eq. (\ref{eq:hamiltonian-hopping}) over varying interaction ranges, (a) average Pauli weight of the improved mapping for $20$ qubits resulting from simulated annealing with Hadamard and $\CNOT$ gates, compared against conventional mappings, and (b) percent reduction Pauli weight over the best conventional mapping for $10$, $15$, and $20$ qubits. Remarkably, the highest percent reduction is achieved for intermediate hopping ranges where we observe a crossover between the Jordan Wigner mapping and the Bravyi-Kitaev and balanced ternary-tree mappings. This result indicates that the optimized mapping leverages the structure of the problem with finite-ranged hopping.}
\label{fig:range-avg}
\end{figure*}

\begin{figure*}[h!]
    \centering
    \includegraphics[width=\textwidth]{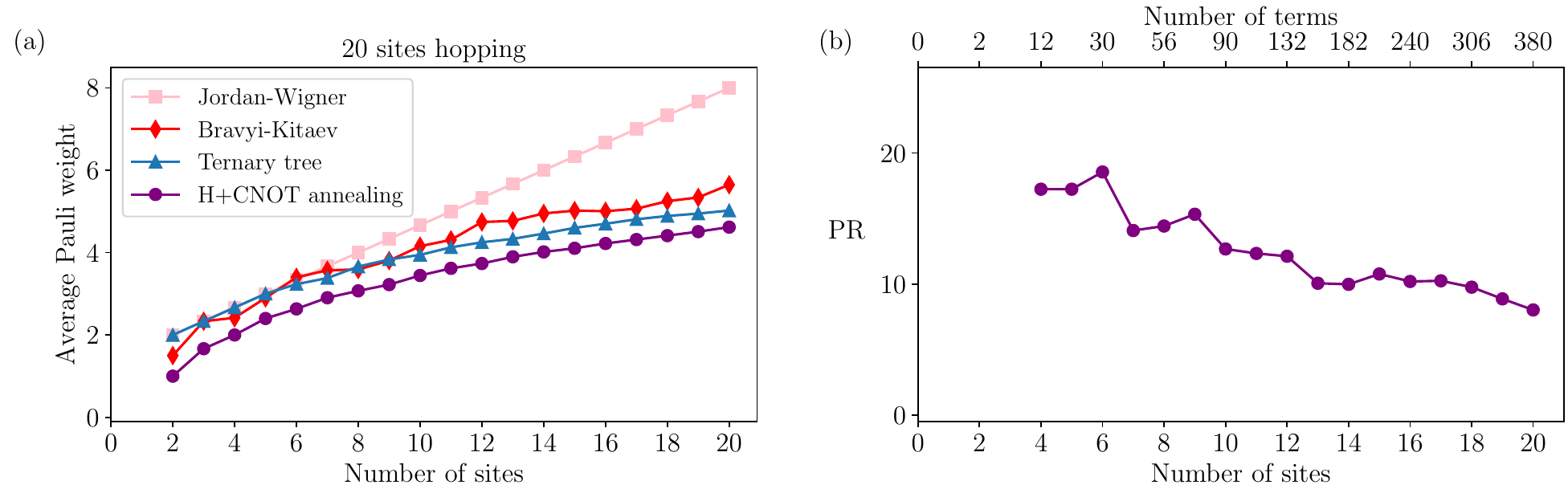}
    \caption{Same as Fig. \ref{fig:range-avg} for all-to-all hopping models, but as a function of system size. The top $x$-axis indicates the number of Hamiltonian terms associated with the bottom $x$-axis.}
    \label{fig:algo-compare}
\end{figure*}

\subsection{1D and 2D hopping models}\label{sec: fermion lattice}
Many fermionic systems in physics are defined on lattices, where the fermionic degrees of freedom associated with each node of the lattice are independent from one another. Moreover, the lattice is commonly embedded in Euclidean space, which provides a natural distance between the lattice nodes. In this setting, the characteristic distance between the nodes that are connected to one-another through hopping or interaction terms (see Section \ref{sec:unitary}) determine the range of hopping or interactions, respectively.

For the sake of a simple exposure to our results, we first consider Hamiltonians which are linear combinations of only hopping terms defined on a lattice. These models provide a convenient benchmark to compare the fermion-to-qubit mappings resulting from our optimization approach to conventional state-of-the-art mappings. Surprisingly, even in the case of simple hopping models, our approach substantially improve the average Pauli weight of conventional fermion-to-qubit mappings, including that of balanced ternary-trees. In the following, we report the improvements resulting from the optimized fermion-to-qubit mappings through the percentage reduction PR$_{\rm conv}$ of the average Pauli weight compared to a conventional mapping $\Phi_{\rm conv}$, defined as
\begin{align}
\mathrm{PR}_{\rm conv} \equiv 1 - \frac{{\rm wt}(\Phi_{\rm opt}(H))}{{\rm wt}(\Phi_{\rm conv}(H))}, 
\end{align}
where $\Phi_{\rm opt}$ denotes the optimized fermion-to-qubit mapping. In the following, we consider $\mathrm{conv}\in \{JWT,BKT, \text{balanced ternary}\}$, and only report the lowest value obtained, denoted ${\rm PR}$.  

Figure \ref{fig:range-avg} shows the results for the hopping model on a one-dimensional lattice with $20$ qubits and varying ranges of hopping terms. Formally, we consider the Hamiltonians of the form
\begin{equation}
\label{eq:hamiltonian-hopping}
H = \sum^N_{0 < \abs{i - j} \le r} a_i^\dagger a_j,
\end{equation}
for various hopping range $0 < r < 20$. Note that we set all coefficients of hopping terms equal to 1. This simplification does not effect our results, because we are using only Clifford unitaries that map one Pauli string to another. 

We observe that given any $r$, the optimized fermion-to-qubit mapping outperforms all conventional mappings. As expected, the Jordan-Wigner mapping is very close to our optimized mapping (all numerical results use the gate set $\mathcal{S}_{CH}$) for short-ranged hopping Hamiltonians. However, for $r>4$, the optimized mapping clearly outperforms all state-of-the-art mappings, including that given by the \emph{optimal ternary-tree mapping}, which provably minimizes the average Pauli weight over all single creation and annihilation operators (see Section \ref{sec:single-ops} for an introduction to ternary-tree mappings). Interestingly, the improvements that we obtain through heuristic optimization peaks around $r=6$ independently from the number of sites in the system. This observation highlights that the optimal mapping leverages a structure of the one-dimensional hopping Hamiltonian with intermediate complexity. Similarly, we see that in the case of all-to-all coupling systems (i.e., $r=N-1$), where we expect to have less structure, the performance of the optimized mapping is the largest around $6$ sites and deteriorates with the increasing number of sites (see Figure \ref{fig:algo-compare} ). Finally, we highlight that the optimized mappings satisfy a necessary condition for ternary-tree mappings: the qubit representation of each single fermion operator anticommute with one another only on a single qubit.

\begin{figure*}[h!]
\includegraphics[width=\textwidth]{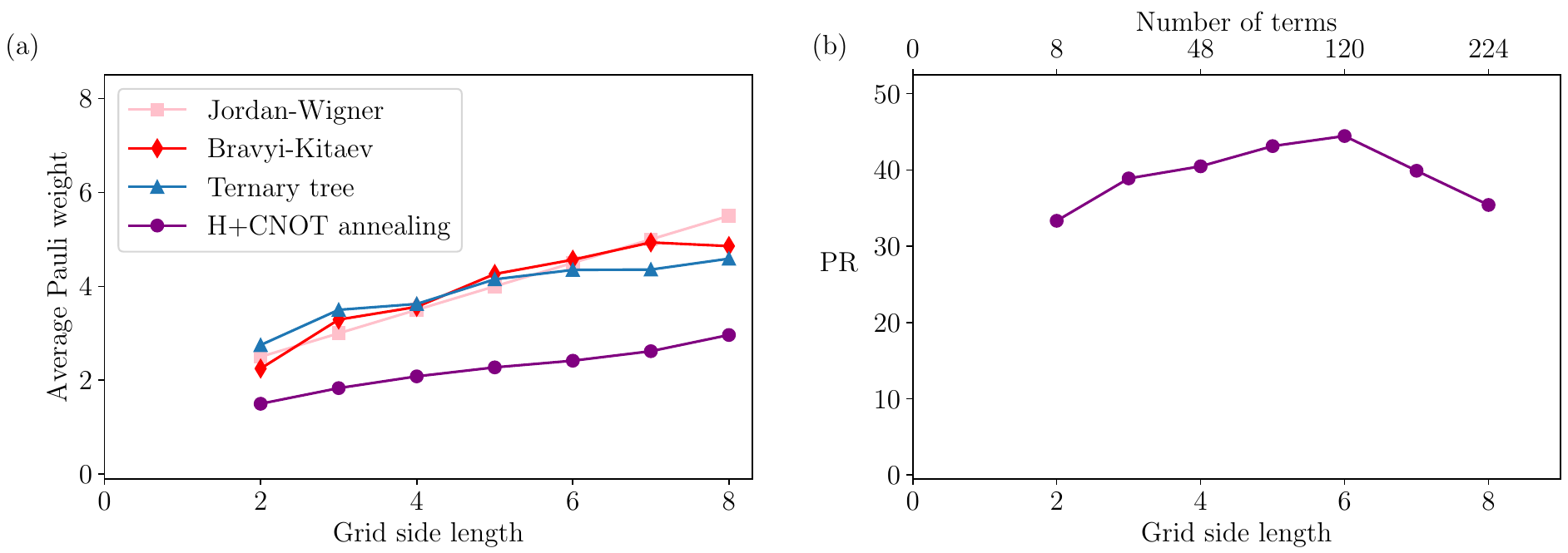}
\caption{Same as Fig. \ref{fig:range-avg} for nearest-neighbor hopping model on a 2D square lattice, as a function of the linear system size (bottom $x$-axis).}
\label{fig:square-avg}
\end{figure*}
\begin{figure*}[h!]
\includegraphics[width=\textwidth]{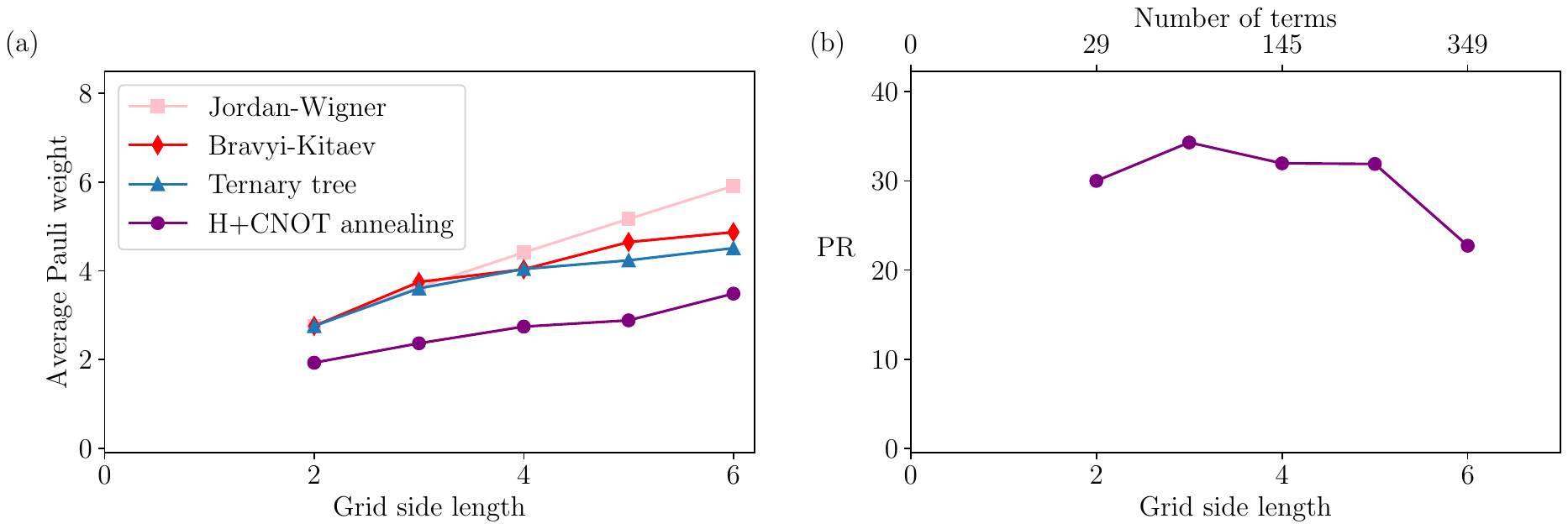}
\caption{Same as Fig. \ref{fig:range-avg} for an interacting 2D Hubbard model on a square lattice, as a function of the linear system size (bottom $x$-axis).}
\label{fig:hubbard-avg}
\end{figure*}

We conclude this subsection with the following interesting observation. The optimized mapping for the nearest-neighbor hopping Hamiltonian (i.e., $r=1$) is that our optimization protocol re-discovers a dual-representation of the nearest-neighbor hopping Hamiltonian  \cite{fradkin2013field}, even in the absence of any externally introduced bias. 
In particular, it is easy to verify that for a one-dimensional lattice with an even number of sites, the optimized mapping results in a qubit Hamiltonian of the form 
\begin{align}
    H_q &= X_1+ \sum_{i=1}^{\frac{N}{2}-1}\left( Z_i Z_{i+1} + X_{i+1}\right)\nonumber\\ 
    &+X_{\frac{N}{2}+1}
    +\sum_{i=\frac{N}{2}+1}^{N-1} \left(Z_i Z_{i+1} + X_{i+1}\right). \label{eq: TFI even}
\end{align}
We note that the average Pauli weight of the representation in Eq. ($\ref{eq: TFI even}$) is $25\%$ lower than the qubit Hamiltonian resulting from a JW mapping $H^{r=1}_{\rm JW}\equiv \sum_{i=1} \left( X_i X_{i+1} + Y_i Y_{i+1} \right)$.

\subsubsection{2D models}

Next, we consider nearest-neighbor hopping models on two-dimensional lattices. It is well-known that Bravyi-Kitaev outperforms Jordan-Wigner in higher dimensions \cite{ultrafast2024obrien}, due to the longer Jordan-Wigner strings on a linearly enumerated lattice. Recent work in Ref. \cite{Chiew_2023} considered improving the average Pauli weight associated with a Jordan-Wigner mapping by optimizing the enumeration of fermionic modes.
However, this approach re-discovers the Mitchison-Durbin ordering \cite{Mitchison1986} to achieve only a 13.9\% reduction in Pauli weight compared to the naive Z or S shaped enumerations and falls behind Bravyi-Kitaev for larger lattices. We find that our optimization approach successfully finds mappings that outperform those introduced in Ref. \cite{Chiew_2023}. 

Figure \ref{fig:square-avg} shows the average Pauli weight comparison between the optimized mapping and the conventional mappings.  We note that for all results shown in Fig. \ref{fig:square-avg}, the initial enumeration of the sites is one that simply snakes through the lattice sites. Remarkably, our optimization far outperforms both the optimal enumeration of Jordan-Wigner \emph{and} the Bravyi-Kitaev mappings, achieving $\approx 35\%$  reduction for 8 by 8 lattices.
However, our approach is still limited by feasibility of numerical optimization. For 9 by 9 grids, the number of qubits is already quite large, and the results suggest that we are approaching the computational limit of our optimization capabilities. By $100$ qubits, little reduction can be obtained in reasonable time. Nonetheless, we emphasize that further analytical understanding of the mappings found through our optimization strategy promises to improve the applicability of our results to larger two-dimensional lattices.

\subsection{2D Hubbard model}

Given the promising advantage of the optimized fermion-to-qubit mapping in two dimensions, we apply our optimization method to the 2D Hubbard model. Our results show that the introduction of on-site density-density interactions between different spin components leads to mappings that retain the percent reduction in average Pauli weight, compared to the Pauli weight associated with an optimal ternary-tree mapping. However, because the number of terms for the spinful 2D Hubbard model is approximately 5/2 times that of spinless 2D model with only hopping terms, it becomes more difficult to optimize mapping for larger systems.
As shown in Figure $\ref{fig:hubbard-avg}$, our optimized mapping outperforms the optimal ternary mapping, which is the most efficient conventional mapping for all system sizes we have considered.
Our results show that a more than a $20\%$ reduction of the average Pauli weight can be achieved for two-dimensional square grids with side length $L\leq 6$, although we already observe that the performance of the optimized mapping is decreasing in this regime. It is also important to emphasize that for a grid side length 6, the number of terms in the Hubbard model is $349$. 
%Finally shown in Figure \ref{fig:hubbard-avg}, the optimal ternary mapping is the most efficient mapping in comparison to the other conventional fermion-to-qubit mappings for all system sizes we have considered.
Moreover, we observe that except for the largest lattice instances, the optimized mappings obey the necessary condition for ternary-tree mappings, that the qubit representation of each single fermion operator anticommute with one another only on a single qubit \cite{Jiang_2020}. While we cannot definitively determine that the optimal fermion-to-qubit mappings for the two-dimensional hopping Hamiltonian is a ternary-tree, we suspect that increasing the time for optimization will result in optimized mappings that are ternary-trees. These results highlight the importance of developing tools to analyze the mappings resulting from our optimization scheme. Determining the regular structures that may arise in our optimized mappings will allow us to design fermion-to-qubit mappings that perform well on extended two-dimensional models.

%%%%%%%%%%%%%%%%%%%%%%%%%%%%%%%%%

\subsection{Hydrogen chains}
\label{sec:hydrogen}
In this section, we study the electronic structure Hamiltonian for linear chains of $n_{\rm atom}$ hydrogen atoms spaced $0.735$\r{A} apart ( e.g., $n_{\rm atom} = 2$, corresponds to a diatomic hydrogen molecule). The Hamiltonians are generated using Qiskit Nature with PySCF, using the STO-3G basis. This has the full form of (\ref{eq:elec-struc-ham}) with both hopping and interaction terms, and the coefficients $h_{ij}$ and $h_{ijkl}$ are computed via classical integration. Including spin, there are $2n_{\rm atom}$ spin orbitals (i.e., $n = 2n_{\rm atom}$ qubits). Because of the long-range interactions between different Hydrogen atoms, the system can be considered as an interacting version of the all-to-all coupling Hamiltonian studies in Section \ref{sec: fermion lattice}. Besides, unlike the interactions in the 2D Hubbard model discussed in the last subsection, the Hamiltonian of the Hydrogen chain includes interaction terms of the form $a_i^{\dagger}a_j^{\dagger}a_ka_l$, where all four indices are unique.

\begin{figure*}[t]
\centering
\includegraphics[width=\textwidth]{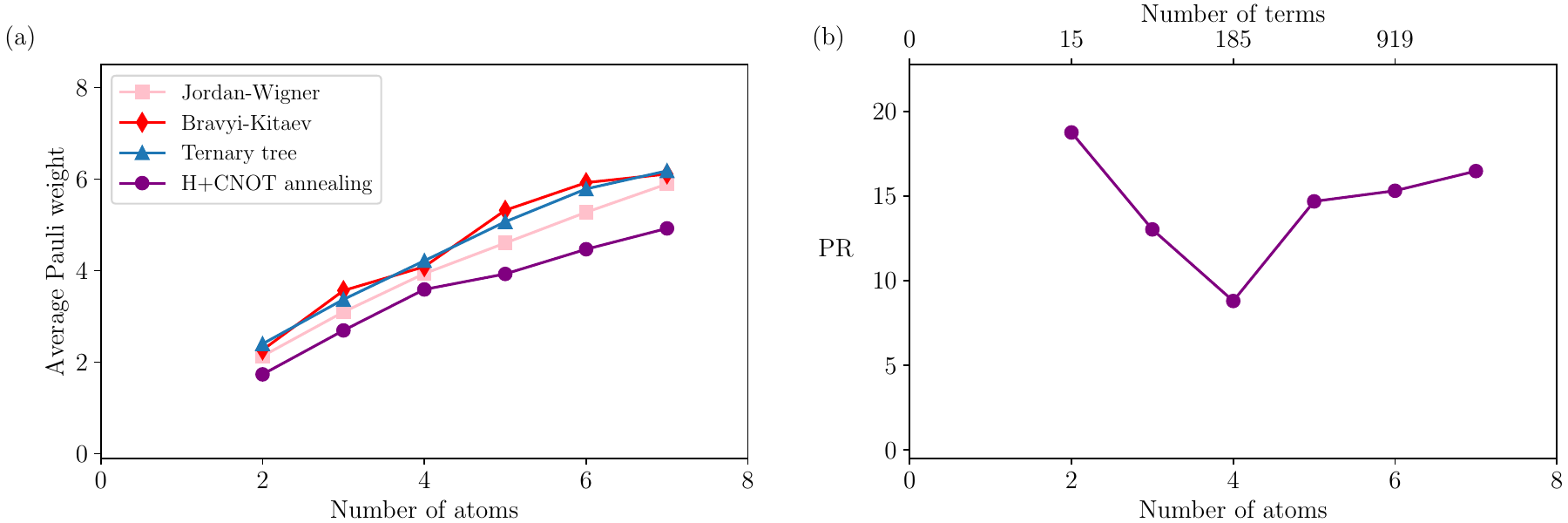}
\caption{Same as Fig. \ref{fig:range-avg} for a linear chain of hydrogen atoms, as a function of the number of Hydrogen atoms (bottom $x$-axis).}
\label{fig:hydrogen-orig}
\end{figure*}

Figure \ref{fig:hydrogen-orig} shows the optimized result. We observe a consistent reduction in Pauli weight by around 15\% compared to the best performing of three conventional mappings, namely JW, BKT, and optimal ternary-tree mappings.
The most important difference between the optimized mappings found for the hydrogen chain Hamiltonian compared to Hamiltonians studied in the previous sections is that the former cannot be expressed as a ternary-tree mapping, even for small system sizes, where we expect the optimization protocol to converge to the global minimum.  
To validate this insight, we employed brute force search for the smallest system sizes, we show that our mapping beats \emph{any} ternary-tree counterpart.
In particular, we find that for $n=4$, none of the approximately $1.5 \times 10^6$ ternary-trees outperform the mapping resulted from the numerical optimization over the circuits constructed with $\mathcal{S}_{\rm CH}$ gate set. The best performing ternary-tree results in an average Pauli weight $\approx 18\%$ higher than that of the optimal mapping.

Remarkably, success of the optimal mapping obtained by applying gates in $\mathcal{S}_{\rm CH}$ to an initial JW mapping can be attributed specifically to the presence of interaction terms that have the form $a_i^\dagger a_j^\dagger a_ka_l$ with $i,j,k,l$ are all distinct. Indeed, the advantage over ternary-trees is present even if we only consider a single exchange interaction term $a_1^\dagger a_2^\dagger a_3 a_4 + h.c.$ Then, the optimal outperforms any ternary-tree mapping by at least $20\%$.

%%%%%%%%%%%%%%%%%%%%%%%%%%%%%%%%%%%%%%%%%%%
\section{Conclusions and Outlook}
\label{sec:conclusion}

To summarize, we proposed a general framework to design new fermion-to-qubit mappings parametrized through Clifford circuits. We demonstrated that our optimization framework for fermion-to-qubit mappings 
can incorporate and leverage the detailed information of underlying fermionic Hamiltonian, thereby allowing us to improve on mappings that only target the Pauli weights of   
individual fermionic creation ($a_i^\dagger$) and annihilation ($a_i$) operators.
We proved that the Clifford ansatz can i) express all optimal ternary-tree mappings for individual fermionic creation ($a_i^\dagger$) and annihilation ($a_i$) operators and ii) improve upon the ternary-tree mapping even when the Hamiltonian is a single interaction term of the form $a^{\dagger}_{i}a^{\dagger}_{j}a_k a_k$.

From a more practical perspective, we demonstrated the advantage of using fermion-to-qubit mappings obtained through our optimization scheme for a  
wide range of Hamiltonians, including one-dimensional hopping and two-dimensional interacting fermionic lattice models, as well as \textit{ab initio} electronic structure Hamiltonians of hydrogen chains. We observed consistent reduction in average Pauli weight in comparison to optimal ternary-tree mappings. 
Last but not least, despite that the run time of the optimization algorithm can be challenging to quantify theoretically, we demonstrate up to $40\%$ reductions in Pauli weight using only modest computational resources over three days.   
Our work provides a general framework for designing more efficient fermion-to-qubit mappings, opening up a practical way of improving the performance of quantum computers in tackling tasks that involve fermionic degrees of freedom \cite{bravyi2002fermionic}.

Despite these accomplishments, a few important open questions remain.
Although our cost function incorporates structures of the underlying fermionic Hamiltonian, it does not include interaction strength in the Hamiltonian or qubit hardware connectivity. These fine-grained information can be incorporated in the cost function to further tailor the mapping towards a specific quantum hardware and/or Hamiltonians of interest. Furthermore, depending on applications, different quantum algorithms, such as finding the ground state \emph{or} simulating the dynamics of the fermionic system, may require some fermionic terms or gates to be optimized more than others. In this respect, we expect that using our framework together with conventional transpilers \cite{javadi2024quantum} would be beneficial. Finally, we expect the impact of optimization based design of fermion-to-qubit mappings to be significantly increased when used together with analytical tools that can uncover the regularities of the optimized mappings. We hope that the ability to further interpret the optimized mapping will pave the way to design customized mappings for intermediate and large scale fermionic systems.

\begin{acknowledgement}
We thank Luke Schaeffer and Joseph Carolan for helpful discussions on ternary-trees and suggesting the theory in Section \ref{sec:single-ops}, as well as Pengfei Zhang on the equivalence between XY model and two disconnected Ising chains. J.Y. was supported in part by the DoE ASCR Quantum Testbed Pathfinder program (awards No. DE-SC0019040 and No. DE-SC0024220), DARPA SAVaNT ADVENT and NSF QLCI (award No. OMA-2120757). Y.L. was supported in part by the U.S. Department of Energy, Office of Science, under contract number DE-SC0025384.

\end{acknowledgement}

\bibstyle{achemso}
%\bibliography{papers}

\begin{mcitethebibliography}{53}
\providecommand*\natexlab[1]{#1}
\providecommand*\mciteSetBstSublistMode[1]{}
\providecommand*\mciteSetBstMaxWidthForm[2]{}
\providecommand*\mciteBstWouldAddEndPuncttrue
  {\def\EndOfBibitem{\unskip.}}
\providecommand*\mciteBstWouldAddEndPunctfalse
  {\let\EndOfBibitem\relax}
\providecommand*\mciteSetBstMidEndSepPunct[3]{}
\providecommand*\mciteSetBstSublistLabelBeginEnd[3]{}
\providecommand*\EndOfBibitem{}
\mciteSetBstSublistMode{f}
\mciteSetBstMaxWidthForm{subitem}{(\alph{mcitesubitemcount})}
\mciteSetBstSublistLabelBeginEnd
  {\mcitemaxwidthsubitemform\space}
  {\relax}
  {\relax}

\bibitem[Schiffer \latin{et~al.}(1995)Schiffer, Ramirez, Bao, and
  Cheong]{Schiffer_PRL_1995}
Schiffer,~P.; Ramirez,~A.~P.; Bao,~W.; Cheong,~S.-W. Low Temperature
  Magnetoresistance and the Magnetic Phase Diagram of
  ${\mathrm{La}}_{1\ensuremath{-}\mathit{x}}{\mathrm{Ca}}_{\mathit{x}}{\mathrm{MnO}}_{3}$.
  \emph{Phys. Rev. Lett.} \textbf{1995}, \emph{75}, 3336--3339\relax
\mciteBstWouldAddEndPuncttrue
\mciteSetBstMidEndSepPunct{\mcitedefaultmidpunct}
{\mcitedefaultendpunct}{\mcitedefaultseppunct}\relax
\EndOfBibitem
\bibitem[Jarrell \latin{et~al.}(2001)Jarrell, Maier, Hettler, and
  Tahvildarzadeh]{Jarrell_EPL_2001}
Jarrell,~M.; Maier,~T.; Hettler,~M.~H.; Tahvildarzadeh,~A.~N. Phase diagram of
  the Hubbard model: Beyond the dynamical mean field. \emph{EPL (Europhysics
  Letters)} \textbf{2001}, \emph{56}, 563\relax
\mciteBstWouldAddEndPuncttrue
\mciteSetBstMidEndSepPunct{\mcitedefaultmidpunct}
{\mcitedefaultendpunct}{\mcitedefaultseppunct}\relax
\EndOfBibitem
\bibitem[Bohn \latin{et~al.}(2017)Bohn, Rey, and Ye]{bohn2017cold}
Bohn,~J.~L.; Rey,~A.~M.; Ye,~J. Cold molecules: Progress in quantum engineering
  of chemistry and quantum matter. \emph{Science} \textbf{2017}, \emph{357},
  1002--1010\relax
\mciteBstWouldAddEndPuncttrue
\mciteSetBstMidEndSepPunct{\mcitedefaultmidpunct}
{\mcitedefaultendpunct}{\mcitedefaultseppunct}\relax
\EndOfBibitem
\bibitem[Balakrishnan(2016)]{balakrishnan2016perspective}
Balakrishnan,~N. Perspective: Ultracold molecules and the dawn of cold
  controlled chemistry. \emph{J. Chem. Phys.} \textbf{2016}, \emph{145},
  150901\relax
\mciteBstWouldAddEndPuncttrue
\mciteSetBstMidEndSepPunct{\mcitedefaultmidpunct}
{\mcitedefaultendpunct}{\mcitedefaultseppunct}\relax
\EndOfBibitem
\bibitem[Feynman(2018)]{feynman2018simulating}
Feynman,~R.~P. \emph{Feynman and computation}; cRc Press, 2018; pp
  133--153\relax
\mciteBstWouldAddEndPuncttrue
\mciteSetBstMidEndSepPunct{\mcitedefaultmidpunct}
{\mcitedefaultendpunct}{\mcitedefaultseppunct}\relax
\EndOfBibitem
\bibitem[Lloyd(1996)]{lloyd1996universal}
Lloyd,~S. Universal quantum simulators. \emph{Science} \textbf{1996},
  \emph{273}, 1073--1078\relax
\mciteBstWouldAddEndPuncttrue
\mciteSetBstMidEndSepPunct{\mcitedefaultmidpunct}
{\mcitedefaultendpunct}{\mcitedefaultseppunct}\relax
\EndOfBibitem
\bibitem[Bharti \latin{et~al.}(2022)Bharti, Cervera-Lierta, Kyaw, Haug,
  Alperin-Lea, Anand, Degroote, Heimonen, Kottmann, Menke, Mok, Sim, Kwek, and
  Aspuru-Guzik]{bharti2022noisy}
Bharti,~K.; Cervera-Lierta,~A.; Kyaw,~T.~H.; Haug,~T.; Alperin-Lea,~S.;
  Anand,~A.; Degroote,~M.; Heimonen,~H.; Kottmann,~J.~S.; Menke,~T.;
  Mok,~W.-K.; Sim,~S.; Kwek,~L.-C.; Aspuru-Guzik,~A. Noisy intermediate-scale
  quantum algorithms. \emph{Rev. Mod. Phys.} \textbf{2022}, \emph{94},
  015004\relax
\mciteBstWouldAddEndPuncttrue
\mciteSetBstMidEndSepPunct{\mcitedefaultmidpunct}
{\mcitedefaultendpunct}{\mcitedefaultseppunct}\relax
\EndOfBibitem
\bibitem[Childs \latin{et~al.}(2018)Childs, Maslov, Nam, Ross, and
  Su]{childs2018toward}
Childs,~A.~M.; Maslov,~D.; Nam,~Y.; Ross,~N.~J.; Su,~Y. Toward the first
  quantum simulation with quantum speedup. \emph{Proceedings of the National
  Academy of Sciences} \textbf{2018}, \emph{115}, 9456--9461\relax
\mciteBstWouldAddEndPuncttrue
\mciteSetBstMidEndSepPunct{\mcitedefaultmidpunct}
{\mcitedefaultendpunct}{\mcitedefaultseppunct}\relax
\EndOfBibitem
\bibitem[Cerezo \latin{et~al.}(2021)Cerezo, Arrasmith, Babbush, Benjamin, Endo,
  Fujii, McClean, Mitarai, Yuan, Cincio, \latin{et~al.}
  others]{cerezo2021variational}
Cerezo,~M.; Arrasmith,~A.; Babbush,~R.; Benjamin,~S.~C.; Endo,~S.; Fujii,~K.;
  McClean,~J.~R.; Mitarai,~K.; Yuan,~X.; Cincio,~L., \latin{et~al.}
  Variational quantum algorithms. \emph{Nature Reviews Physics} \textbf{2021},
  \emph{3}, 625--644\relax
\mciteBstWouldAddEndPuncttrue
\mciteSetBstMidEndSepPunct{\mcitedefaultmidpunct}
{\mcitedefaultendpunct}{\mcitedefaultseppunct}\relax
\EndOfBibitem
\bibitem[Bravyi and Kitaev(2002)Bravyi, and Kitaev]{bravyi2002fermionic}
Bravyi,~S.~B.; Kitaev,~A.~Y. Fermionic quantum computation. \emph{Annals of
  Physics} \textbf{2002}, \emph{298}, 210--226\relax
\mciteBstWouldAddEndPuncttrue
\mciteSetBstMidEndSepPunct{\mcitedefaultmidpunct}
{\mcitedefaultendpunct}{\mcitedefaultseppunct}\relax
\EndOfBibitem
\bibitem[Li \latin{et~al.}(2022)Li, Wu, Shi, Javadi-Abhari, Ding, and
  Xie]{paulihedral}
Li,~G.; Wu,~A.; Shi,~Y.; Javadi-Abhari,~A.; Ding,~Y.; Xie,~Y. Paulihedral: a
  generalized block-wise compiler optimization framework for Quantum simulation
  kernels. Proceedings of the 27th ACM International Conference on
  Architectural Support for Programming Languages and Operating Systems. New
  York, NY, USA, 2022; p 554–569\relax
\mciteBstWouldAddEndPuncttrue
\mciteSetBstMidEndSepPunct{\mcitedefaultmidpunct}
{\mcitedefaultendpunct}{\mcitedefaultseppunct}\relax
\EndOfBibitem
\bibitem[Holmes \latin{et~al.}(2022)Holmes, Sharma, Cerezo, and
  Coles]{holmes2022connecting}
Holmes,~Z.; Sharma,~K.; Cerezo,~M.; Coles,~P.~J. Connecting Ansatz
  Expressibility to Gradient Magnitudes and Barren Plateaus. \emph{PRX Quantum}
  \textbf{2022}, \emph{3}, 010313\relax
\mciteBstWouldAddEndPuncttrue
\mciteSetBstMidEndSepPunct{\mcitedefaultmidpunct}
{\mcitedefaultendpunct}{\mcitedefaultseppunct}\relax
\EndOfBibitem
\bibitem[Huang \latin{et~al.}(2020)Huang, Kueng, and
  Preskill]{huang2020predicting}
Huang,~H.-Y.; Kueng,~R.; Preskill,~J. Predicting many properties of a quantum
  system from very few measurements. \emph{Nature Physics} \textbf{2020},
  \emph{16}, 1050--1057\relax
\mciteBstWouldAddEndPuncttrue
\mciteSetBstMidEndSepPunct{\mcitedefaultmidpunct}
{\mcitedefaultendpunct}{\mcitedefaultseppunct}\relax
\EndOfBibitem
\bibitem[McClean \latin{et~al.}(2019)McClean, Sung, Kivlichan, Cao, Dai, Fried,
  Gidney, Gimby, Gokhale, Häner, Hardikar, Havlíček, Higgott, Huang, Izaac,
  Jiang, Liu, McArdle, Neeley, O'Brien, O'Gorman, Ozfidan, Radin, Romero,
  Rubin, Sawaya, Setia, Sim, Steiger, Steudtner, Sun, Sun, Wang, Zhang, and
  Babbush]{mcclean2019openfermion}
McClean,~J.~R.; Sung,~K.~J.; Kivlichan,~I.~D.; Cao,~Y.; Dai,~C.; Fried,~E.~S.;
  Gidney,~C.; Gimby,~B.; Gokhale,~P.; Häner,~T.; Hardikar,~T.; Havlíček,~V.;
  Higgott,~O.; Huang,~C.; Izaac,~J.; Jiang,~Z.; Liu,~X.; McArdle,~S.;
  Neeley,~M.; O'Brien,~T.; O'Gorman,~B.; Ozfidan,~I.; Radin,~M.~D.; Romero,~J.;
  Rubin,~N.; Sawaya,~N. P.~D.; Setia,~K.; Sim,~S.; Steiger,~D.~S.;
  Steudtner,~M.; Sun,~Q.; Sun,~W.; Wang,~D.; Zhang,~F.; Babbush,~R.
  OpenFermion: The Electronic Structure Package for Quantum Computers.
  2019\relax
\mciteBstWouldAddEndPuncttrue
\mciteSetBstMidEndSepPunct{\mcitedefaultmidpunct}
{\mcitedefaultendpunct}{\mcitedefaultseppunct}\relax
\EndOfBibitem
\bibitem[{The Qiskit Nature developers and contributors}(2023)]{qiskit_nature}
{The Qiskit Nature developers and contributors}, Qiskit Nature 0.6.0. 2023;
  \url{https://doi.org/10.5281/zenodo.7828768}\relax
\mciteBstWouldAddEndPuncttrue
\mciteSetBstMidEndSepPunct{\mcitedefaultmidpunct}
{\mcitedefaultendpunct}{\mcitedefaultseppunct}\relax
\EndOfBibitem
\bibitem[Seeley \latin{et~al.}(2012)Seeley, Richard, and Love]{Seeley_2012}
Seeley,~J.~T.; Richard,~M.~J.; Love,~P.~J. The Bravyi-Kitaev transformation for
  quantum computation of electronic structure. \emph{The Journal of Chemical
  Physics} \textbf{2012}, \emph{137}\relax
\mciteBstWouldAddEndPuncttrue
\mciteSetBstMidEndSepPunct{\mcitedefaultmidpunct}
{\mcitedefaultendpunct}{\mcitedefaultseppunct}\relax
\EndOfBibitem
\bibitem[Tranter \latin{et~al.}(2018)Tranter, Love, Mintert, and
  Coveney]{Tranter_2018}
Tranter,~A.; Love,~P.~J.; Mintert,~F.; Coveney,~P.~V. A Comparison of the
  Bravyi{\textendash}Kitaev and Jordan{\textendash}Wigner Transformations for
  the Quantum Simulation of Quantum Chemistry. \emph{Journal of Chemical Theory
  and Computation} \textbf{2018}, \emph{14}, 5617--5630\relax
\mciteBstWouldAddEndPuncttrue
\mciteSetBstMidEndSepPunct{\mcitedefaultmidpunct}
{\mcitedefaultendpunct}{\mcitedefaultseppunct}\relax
\EndOfBibitem
\bibitem[Jiang \latin{et~al.}(2020)Jiang, Kalev, Mruczkiewicz, and
  Neven]{Jiang_2020}
Jiang,~Z.; Kalev,~A.; Mruczkiewicz,~W.; Neven,~H. Optimal fermion-to-qubit
  mapping via ternary trees with applications to reduced quantum states
  learning. \emph{Quantum} \textbf{2020}, \emph{4}, 276\relax
\mciteBstWouldAddEndPuncttrue
\mciteSetBstMidEndSepPunct{\mcitedefaultmidpunct}
{\mcitedefaultendpunct}{\mcitedefaultseppunct}\relax
\EndOfBibitem
\bibitem[Bravyi \latin{et~al.}(2017)Bravyi, Gambetta, Mezzacapo, and
  Temme]{bravyi2017tapering}
Bravyi,~S.; Gambetta,~J.~M.; Mezzacapo,~A.; Temme,~K. Tapering off qubits to
  simulate fermionic Hamiltonians. \emph{arXiv:1701.08213} \textbf{2017},
  \relax
\mciteBstWouldAddEndPunctfalse
\mciteSetBstMidEndSepPunct{\mcitedefaultmidpunct}
{}{\mcitedefaultseppunct}\relax
\EndOfBibitem
\bibitem[Fischer and Gunlycke(2019)Fischer, and Gunlycke]{fischer2019symmetry}
Fischer,~S.~A.; Gunlycke,~D. Symmetry configuration mapping for representing
  quantum systems on quantum computers. \emph{arXiv preprint arXiv:1907.01493}
  \textbf{2019}, \relax
\mciteBstWouldAddEndPunctfalse
\mciteSetBstMidEndSepPunct{\mcitedefaultmidpunct}
{}{\mcitedefaultseppunct}\relax
\EndOfBibitem
\bibitem[Verstraete and Cirac(2005)Verstraete, and
  Cirac]{verstraete2005mapping}
Verstraete,~F.; Cirac,~J.~I. Mapping local Hamiltonians of fermions to local
  Hamiltonians of spins. \emph{Journal of Statistical Mechanics: Theory and
  Experiment} \textbf{2005}, \emph{2005}, P09012\relax
\mciteBstWouldAddEndPuncttrue
\mciteSetBstMidEndSepPunct{\mcitedefaultmidpunct}
{\mcitedefaultendpunct}{\mcitedefaultseppunct}\relax
\EndOfBibitem
\bibitem[Whitfield \latin{et~al.}(2016)Whitfield, Havl{\'\i}{\v{c}}ek, and
  Troyer]{whitfield2016local}
Whitfield,~J.~D.; Havl{\'\i}{\v{c}}ek,~V.; Troyer,~M. Local spin operators for
  fermion simulations. \emph{Physical Review A} \textbf{2016}, \emph{94},
  030301\relax
\mciteBstWouldAddEndPuncttrue
\mciteSetBstMidEndSepPunct{\mcitedefaultmidpunct}
{\mcitedefaultendpunct}{\mcitedefaultseppunct}\relax
\EndOfBibitem
\bibitem[Steudtner and Wehner(2018)Steudtner, and Wehner]{Steudtner_2018}
Steudtner,~M.; Wehner,~S. Fermion-to-qubit mappings with varying resource
  requirements for quantum simulation. \emph{New Journal of Physics}
  \textbf{2018}, \emph{20}, 063010\relax
\mciteBstWouldAddEndPuncttrue
\mciteSetBstMidEndSepPunct{\mcitedefaultmidpunct}
{\mcitedefaultendpunct}{\mcitedefaultseppunct}\relax
\EndOfBibitem
\bibitem[Setia \latin{et~al.}(2019)Setia, Bravyi, Mezzacapo, and
  Whitfield]{Setia_2019}
Setia,~K.; Bravyi,~S.; Mezzacapo,~A.; Whitfield,~J.~D. Superfast encodings for
  fermionic quantum simulation. \emph{Physical Review Research} \textbf{2019},
  \emph{1}\relax
\mciteBstWouldAddEndPuncttrue
\mciteSetBstMidEndSepPunct{\mcitedefaultmidpunct}
{\mcitedefaultendpunct}{\mcitedefaultseppunct}\relax
\EndOfBibitem
\bibitem[Derby \latin{et~al.}(2021)Derby, Klassen, Bausch, and
  Cubitt]{derby2021compact}
Derby,~C.; Klassen,~J.; Bausch,~J.; Cubitt,~T. Compact fermion to qubit
  mappings. \emph{Phys. Rev. B} \textbf{2021}, \emph{104}, 035118\relax
\mciteBstWouldAddEndPuncttrue
\mciteSetBstMidEndSepPunct{\mcitedefaultmidpunct}
{\mcitedefaultendpunct}{\mcitedefaultseppunct}\relax
\EndOfBibitem
\bibitem[O'Brien and Strelchuk(2024)O'Brien, and
  Strelchuk]{ultrafast2024obrien}
O'Brien,~O.; Strelchuk,~S. Ultrafast hybrid fermion-to-qubit mapping.
  \emph{Phys. Rev. B} \textbf{2024}, \emph{109}, 115149\relax
\mciteBstWouldAddEndPuncttrue
\mciteSetBstMidEndSepPunct{\mcitedefaultmidpunct}
{\mcitedefaultendpunct}{\mcitedefaultseppunct}\relax
\EndOfBibitem
\bibitem[Babbush \latin{et~al.}(2017)Babbush, Berry, Sanders, Kivlichan,
  Scherer, Wei, Love, and Aspuru-Guzik]{Babbush_2018}
Babbush,~R.; Berry,~D.~W.; Sanders,~Y.~R.; Kivlichan,~I.~D.; Scherer,~A.;
  Wei,~A.~Y.; Love,~P.~J.; Aspuru-Guzik,~A. Exponentially more precise quantum
  simulation of fermions in the configuration interaction representation.
  \emph{Quantum Science and Technology} \textbf{2017}, \emph{3}, 015006\relax
\mciteBstWouldAddEndPuncttrue
\mciteSetBstMidEndSepPunct{\mcitedefaultmidpunct}
{\mcitedefaultendpunct}{\mcitedefaultseppunct}\relax
\EndOfBibitem
\bibitem[Hastings \latin{et~al.}(2014)Hastings, Wecker, Bauer, and
  Troyer]{hastings2014improving}
Hastings,~M.~B.; Wecker,~D.; Bauer,~B.; Troyer,~M. Improving quantum algorithms
  for quantum chemistry. \emph{arXiv preprint arXiv:1403.1539} \textbf{2014},
  \relax
\mciteBstWouldAddEndPunctfalse
\mciteSetBstMidEndSepPunct{\mcitedefaultmidpunct}
{}{\mcitedefaultseppunct}\relax
\EndOfBibitem
\bibitem[Sun \latin{et~al.}(2023)Sun, Cheng, and Li]{sun2023towards}
Sun,~J.; Cheng,~L.; Li,~W. Towards chemical accuracy with shallow quantum
  circuits: A Clifford-based Hamiltonian engineering approach. \emph{arXiv
  preprint arXiv:2306.12053} \textbf{2023}, \relax
\mciteBstWouldAddEndPunctfalse
\mciteSetBstMidEndSepPunct{\mcitedefaultmidpunct}
{}{\mcitedefaultseppunct}\relax
\EndOfBibitem
\bibitem[Mishmash \latin{et~al.}(2023)Mishmash, Gujarati, Motta, Zhai, Chan,
  and Mezzacapo]{mishmash2023hierarchical}
Mishmash,~R.~V.; Gujarati,~T.~P.; Motta,~M.; Zhai,~H.; Chan,~G. K.-L.;
  Mezzacapo,~A. Hierarchical Clifford transformations to reduce entanglement in
  quantum chemistry wavefunctions. 2023\relax
\mciteBstWouldAddEndPuncttrue
\mciteSetBstMidEndSepPunct{\mcitedefaultmidpunct}
{\mcitedefaultendpunct}{\mcitedefaultseppunct}\relax
\EndOfBibitem
\bibitem[Chien and Klassen(2022)Chien, and Klassen]{chien2022optimizing}
Chien,~R.~W.; Klassen,~J. Optimizing fermionic encodings for both Hamiltonian
  and hardware. \emph{arXiv preprint arXiv:2210.05652} \textbf{2022}, \relax
\mciteBstWouldAddEndPunctfalse
\mciteSetBstMidEndSepPunct{\mcitedefaultmidpunct}
{}{\mcitedefaultseppunct}\relax
\EndOfBibitem
\bibitem[Nys and Carleo(2022)Nys, and Carleo]{nys2022variational}
Nys,~J.; Carleo,~G. Variational solutions to fermion-to-qubit mappings in two
  spatial dimensions. \emph{Quantum} \textbf{2022}, \emph{6}, 833\relax
\mciteBstWouldAddEndPuncttrue
\mciteSetBstMidEndSepPunct{\mcitedefaultmidpunct}
{\mcitedefaultendpunct}{\mcitedefaultseppunct}\relax
\EndOfBibitem
\bibitem[Chiew and Strelchuk(2023)Chiew, and Strelchuk]{Chiew_2023}
Chiew,~M.; Strelchuk,~S. Discovering optimal fermion-qubit mappings through
  algorithmic enumeration. \emph{Quantum} \textbf{2023}, \emph{7}, 1145\relax
\mciteBstWouldAddEndPuncttrue
\mciteSetBstMidEndSepPunct{\mcitedefaultmidpunct}
{\mcitedefaultendpunct}{\mcitedefaultseppunct}\relax
\EndOfBibitem
\bibitem[Miller \latin{et~al.}(2024)Miller, Glos, and
  Zimbor{\'a}s]{miller2024treespilation}
Miller,~A.; Glos,~A.; Zimbor{\'a}s,~Z. Treespilation: Architecture-and
  State-Optimised Fermion-to-Qubit Mappings. \emph{arXiv preprint
  arXiv:2403.03992} \textbf{2024}, \relax
\mciteBstWouldAddEndPunctfalse
\mciteSetBstMidEndSepPunct{\mcitedefaultmidpunct}
{}{\mcitedefaultseppunct}\relax
\EndOfBibitem
\bibitem[Kirkpatrick \latin{et~al.}(1983)Kirkpatrick, Gelatt~Jr, and
  Vecchi]{kirkpatrick1983optimization}
Kirkpatrick,~S.; Gelatt~Jr,~C.~D.; Vecchi,~M.~P. Optimization by simulated
  annealing. \emph{science} \textbf{1983}, \emph{220}, 671--680\relax
\mciteBstWouldAddEndPuncttrue
\mciteSetBstMidEndSepPunct{\mcitedefaultmidpunct}
{\mcitedefaultendpunct}{\mcitedefaultseppunct}\relax
\EndOfBibitem
\bibitem[Gottesman(1998)]{gottesman1998heisenberg}
Gottesman,~D. The Heisenberg representation of quantum computers. \emph{arXiv
  preprint quant-ph/9807006} \textbf{1998}, \relax
\mciteBstWouldAddEndPunctfalse
\mciteSetBstMidEndSepPunct{\mcitedefaultmidpunct}
{}{\mcitedefaultseppunct}\relax
\EndOfBibitem
\bibitem[Altland and Simons(2010)Altland, and Simons]{altland2010condensed}
Altland,~A.; Simons,~B.~D. \emph{Condensed matter field theory}; Cambridge
  university press, 2010\relax
\mciteBstWouldAddEndPuncttrue
\mciteSetBstMidEndSepPunct{\mcitedefaultmidpunct}
{\mcitedefaultendpunct}{\mcitedefaultseppunct}\relax
\EndOfBibitem
\bibitem[Tranter \latin{et~al.}(2015)Tranter, Sofia, Seeley, Kaicher, McClean,
  Babbush, Coveney, Mintert, Wilhelm, and Love]{tranter2015b}
Tranter,~A.; Sofia,~S.; Seeley,~J.; Kaicher,~M.; McClean,~J.; Babbush,~R.;
  Coveney,~P.~V.; Mintert,~F.; Wilhelm,~F.; Love,~P.~J. The Bravyi--Kitaev
  transformation: Properties and applications. \emph{International Journal of
  Quantum Chemistry} \textbf{2015}, \emph{115}, 1431--1441\relax
\mciteBstWouldAddEndPuncttrue
\mciteSetBstMidEndSepPunct{\mcitedefaultmidpunct}
{\mcitedefaultendpunct}{\mcitedefaultseppunct}\relax
\EndOfBibitem
\bibitem[Cross \latin{et~al.}(2019)Cross, Bishop, Sheldon, Nation, and
  Gambetta]{cross2019validating}
Cross,~A.~W.; Bishop,~L.~S.; Sheldon,~S.; Nation,~P.~D.; Gambetta,~J.~M.
  Validating quantum computers using randomized model circuits. \emph{Phys.
  Rev. A} \textbf{2019}, \emph{100}, 032328\relax
\mciteBstWouldAddEndPuncttrue
\mciteSetBstMidEndSepPunct{\mcitedefaultmidpunct}
{\mcitedefaultendpunct}{\mcitedefaultseppunct}\relax
\EndOfBibitem
\bibitem[Aaronson and Gottesman(2004)Aaronson, and
  Gottesman]{aaronson2004improved}
Aaronson,~S.; Gottesman,~D. Improved simulation of stabilizer circuits.
  \emph{Physical Review A} \textbf{2004}, \emph{70}, 052328\relax
\mciteBstWouldAddEndPuncttrue
\mciteSetBstMidEndSepPunct{\mcitedefaultmidpunct}
{\mcitedefaultendpunct}{\mcitedefaultseppunct}\relax
\EndOfBibitem
\bibitem[Nielsen and Chuang(2010)Nielsen, and Chuang]{nielsen2010quantum}
Nielsen,~M.~A.; Chuang,~I.~L. \emph{Quantum computation and quantum
  information}; Cambridge university press, 2010\relax
\mciteBstWouldAddEndPuncttrue
\mciteSetBstMidEndSepPunct{\mcitedefaultmidpunct}
{\mcitedefaultendpunct}{\mcitedefaultseppunct}\relax
\EndOfBibitem
\bibitem[Miller \latin{et~al.}(2023)Miller, Zimbor\'as, Knecht, Maniscalco, and
  Garc\'{\i}a-P\'erez]{bonsai2024miller}
Miller,~A.; Zimbor\'as,~Z.; Knecht,~S.; Maniscalco,~S.; Garc\'{\i}a-P\'erez,~G.
  Bonsai Algorithm: Grow Your Own Fermion-to-Qubit Mappings. \emph{PRX Quantum}
  \textbf{2023}, \emph{4}, 030314\relax
\mciteBstWouldAddEndPuncttrue
\mciteSetBstMidEndSepPunct{\mcitedefaultmidpunct}
{\mcitedefaultendpunct}{\mcitedefaultseppunct}\relax
\EndOfBibitem
\bibitem[Vlasov(2022)]{Vlasov_2022}
Vlasov,~A.~Y. Clifford Algebras, Spin Groups and Qubit Trees. \emph{Quanta}
  \textbf{2022}, \emph{11}, 97–114\relax
\mciteBstWouldAddEndPuncttrue
\mciteSetBstMidEndSepPunct{\mcitedefaultmidpunct}
{\mcitedefaultendpunct}{\mcitedefaultseppunct}\relax
\EndOfBibitem
\bibitem[Cormen \latin{et~al.}(2009)Cormen, Leiserson, Rivest, and Stein]{clrs}
Cormen,~T.~H.; Leiserson,~C.~E.; Rivest,~R.~L.; Stein,~C. \emph{Introduction to
  Algorithms, Third Edition}, 3rd ed.; The MIT Press, 2009\relax
\mciteBstWouldAddEndPuncttrue
\mciteSetBstMidEndSepPunct{\mcitedefaultmidpunct}
{\mcitedefaultendpunct}{\mcitedefaultseppunct}\relax
\EndOfBibitem
\bibitem[Fradkin(2013)]{fradkin2013field}
Fradkin,~E. \emph{Field theories of condensed matter physics}; Cambridge
  University Press, 2013\relax
\mciteBstWouldAddEndPuncttrue
\mciteSetBstMidEndSepPunct{\mcitedefaultmidpunct}
{\mcitedefaultendpunct}{\mcitedefaultseppunct}\relax
\EndOfBibitem
\bibitem[Mitchison and Durbin(1986)Mitchison, and Durbin]{Mitchison1986}
Mitchison,~G.~J.; Durbin,~R. Optimal numberings of an N N array. \emph{Siam
  Journal on Algebraic and Discrete Methods} \textbf{1986}, \emph{7},
  571--582\relax
\mciteBstWouldAddEndPuncttrue
\mciteSetBstMidEndSepPunct{\mcitedefaultmidpunct}
{\mcitedefaultendpunct}{\mcitedefaultseppunct}\relax
\EndOfBibitem
\bibitem[Javadi-Abhari \latin{et~al.}(2024)Javadi-Abhari, Treinish, Krsulich,
  Wood, Lishman, Gacon, Martiel, Nation, Bishop, Cross, \latin{et~al.}
  others]{javadi2024quantum}
Javadi-Abhari,~A.; Treinish,~M.; Krsulich,~K.; Wood,~C.~J.; Lishman,~J.;
  Gacon,~J.; Martiel,~S.; Nation,~P.~D.; Bishop,~L.~S.; Cross,~A.~W.,
  \latin{et~al.}  Quantum computing with Qiskit. \emph{arXiv preprint
  arXiv:2405.08810} \textbf{2024}, \relax
\mciteBstWouldAddEndPunctfalse
\mciteSetBstMidEndSepPunct{\mcitedefaultmidpunct}
{}{\mcitedefaultseppunct}\relax
\EndOfBibitem
\bibitem[van Laarhoven and Aarts(1987)van Laarhoven, and
  Aarts]{vanLaarhoven1987}
van Laarhoven,~P. J.~M.; Aarts,~E. H.~L. \emph{Simulated Annealing: Theory and
  Applications}; Springer Netherlands: Dordrecht, 1987; pp 7--15\relax
\mciteBstWouldAddEndPuncttrue
\mciteSetBstMidEndSepPunct{\mcitedefaultmidpunct}
{\mcitedefaultendpunct}{\mcitedefaultseppunct}\relax
\EndOfBibitem
\bibitem[Bertsimas and Tsitsiklis(1993)Bertsimas, and
  Tsitsiklis]{Bertsimas1993}
Bertsimas,~D.; Tsitsiklis,~J. {Simulated Annealing}. \emph{Statistical Science}
  \textbf{1993}, \emph{8}, 10 -- 15\relax
\mciteBstWouldAddEndPuncttrue
\mciteSetBstMidEndSepPunct{\mcitedefaultmidpunct}
{\mcitedefaultendpunct}{\mcitedefaultseppunct}\relax
\EndOfBibitem
\bibitem[Nikolaev and Jacobson(2010)Nikolaev, and Jacobson]{Nikolaev2010}
Nikolaev,~A.~G.; Jacobson,~S.~H. In \emph{Handbook of Metaheuristics};
  Gendreau,~M., Potvin,~J.-Y., Eds.; Springer US: Boston, MA, 2010; pp
  1--39\relax
\mciteBstWouldAddEndPuncttrue
\mciteSetBstMidEndSepPunct{\mcitedefaultmidpunct}
{\mcitedefaultendpunct}{\mcitedefaultseppunct}\relax
\EndOfBibitem
\bibitem[Russell and Norvig(2009)Russell, and Norvig]{ai-bfs}
Russell,~S.; Norvig,~P. \emph{Artificial Intelligence: A Modern Approach}, 3rd
  ed.; Prentice Hall Press: USA, 2009\relax
\mciteBstWouldAddEndPuncttrue
\mciteSetBstMidEndSepPunct{\mcitedefaultmidpunct}
{\mcitedefaultendpunct}{\mcitedefaultseppunct}\relax
\EndOfBibitem
\bibitem[Graham \latin{et~al.}(1994)Graham, Knuth, and
  Patashnik]{graham1994concrete}
Graham,~R.; Knuth,~D.; Patashnik,~O. \emph{Concrete Mathematics: A Foundation
  for Computer Science}; Addison-Wesley, 1994\relax
\mciteBstWouldAddEndPuncttrue
\mciteSetBstMidEndSepPunct{\mcitedefaultmidpunct}
{\mcitedefaultendpunct}{\mcitedefaultseppunct}\relax
\EndOfBibitem
\end{mcitethebibliography}

\providecommand{\latin}[1]{#1}
\makeatletter
\providecommand{\doi}
  {\begingroup\let\do\@makeother\dospecials
  \catcode`\{=1 \catcode`\}=2 \doi@aux}
\providecommand{\doi@aux}[1]{\endgroup\texttt{#1}}
\makeatother
\providecommand*\mcitethebibliography{\thebibliography}
\csname @ifundefined\endcsname{endmcitethebibliography}
  {\let\endmcitethebibliography\endthebibliography}{}

\appendix

\section{Algorithms} 
\label{sec:algorithms}

In this Appendix, we discuss the details of our optimization algorithm and compare the choices of search algorithm and gate sets. We first present two search algorithms in \ref{subsec:appendix-sa} and  \ref{subsec:appendix-bfs}. We compare these in \ref{sec:comparison} and discuss the choice of Clifford gate sets.

\subsection{Simulated annealing}
\label{subsec:appendix-sa}

Our primary search algorithim is simulated annealing \cite{kirkpatrick1983optimization,vanLaarhoven1987, Bertsimas1993, Nikolaev2010}. Simulated annealing for quantum circuit optimization utilizes a probabilistic strategy of sequentially adding Clifford gates one at a time. The pseudocode of the simulated annealing algorithm we used is given in Algorithm $\ref{alg:sa}$. The algorithm initiates with an initial quantum circuit $U=I$ and systematically explores neighboring solutions by incorporating one Clifford gate at each iteration.
If the updated $U_{\rm C}$ yields a lower cost, it is promptly accepted as the current solution. However, to evade local minima, circuits resulting in higher costs are stochastically accepted, relying on a probability function controlled by the temperature parameter $\beta$. This probability diminishes as the algorithm progresses, gradually narrowing the search towards more promising configurations, mirroring the cooling process. The algorithm concludes when predefined stopping criteria are met, such as reaching the maximum number of iterations or attaining the target temperature. The output is the most optimal quantum circuit discovered during the search.
During the calculation, $\beta$ depends on the number of iterations $t$, and $U_C$ is dynamically updated. Hence, we call $\beta(t)$ the annealing schedule. In Algorithm \ref{alg:sa}, $U_C$ is given by the product of the quantum gates in the sequence $M$. 

The choice of the annealing schedule $\beta(t)$ has significant impacts on the performance \cite{Nikolaev2010}. If the schedule is too aggressive (i.e. grows too fast), the search may get stuck in a local minimum early on. However, if it is too slow, the search may explore far away and take a long time to converge. We use schedules of the general form
\begin{align}
    \beta(t) = \log(c_1 + c_2 t) \cdot \frac{c_3}{C(B)} \cdot \Theta(t-t_{\min})
    \label{eq: beta}
\end{align}
and manually adjust the hyperparameters $c_1$, $c_2$, $c_3$, and $t_{\min}$ for each run. The $\Theta(t-t_{\min})$ term results in random gates being applied for the first $t_{\min}$ steps, as $\beta = 0$ corresponds to infinite temperature. 

In scenarios with an extensive number of gate combinations, this approach permits the algorithm to halt at any point and deliver the best solution found thus far. Although the heuristic nature of simulated annealing does not guarantee global optimality, its efficacy lies in efficiently traversing large solution spaces.

\begin{algorithm*}[t]
    \caption{Simulated annealing for fermionic-qubit mapping}
    \label{alg:sa}
    \raggedright
    \textbf{Input}: Hamiltonian $H$, gate set $\mathcal{S}$, cost function $C$, annealing schedule $\beta$, number of iterations $t_{\max}$ \\
    \textbf{Output}: New Hamiltonian with reduced cost and a (very long) list of moves
    
    \begin{algorithmic}[1]
    \State initialize move list $M \gets [\ ]$
    \For {$t=1$ to $t_{\max}$}
        \State randomly sample a gate $G \in \mathcal{S}$
        \If {with probability $e^{-\beta(t) [C(GBG^\dagger) - C(B)]}$}
            \State accept $G$ by setting $B \gets GBG^\dagger$ and appending $G$ to $M$
        \EndIf
    \EndFor
    \State \Return $(B, M)$
    \end{algorithmic}
\end{algorithm*}

\subsection{Best-first search}
\label{subsec:appendix-bfs}
For completeness, we also include a simple algorithm based on best-first search, a greedy algorithm that explores the lowest-cost option first while storing the next-best options for backtracking.\cite{ai-bfs} The pseudocode for the algorithm that we use is given in Algorithm $\ref{alg:SA}$. 
This algorithm is deterministic and can eventually search all combinations. However, when the number of combinations is large, we have to stop the algorithm in the middle of the calculation and take the best result as the solution at that point.

\begin{algorithm*}[t]
    \caption{Best-first search for fermionic-qubit mapping}
    \label{alg:search}
    \raggedright
    \textbf{Input}: Hamiltonian $H$, gate set $\mathcal{S}$, cost function $\mathcal{C}$ \\
    \textbf{Output}: Sequence of gates from $\mathcal{S}$ which best reduces the cost of $H$
    
    \begin{algorithmic}[1]
    \State initialize best cost $b \gets \mathcal{C}(H)$, best gate sequence $B \gets [\ ]$
    \State initialize priority queue $\mathcal{Q} \gets [(H, B)]$ 
    with $\mathcal{C}(H)$ as priority
    \While {$\mathcal{Q}$ is not empty and not timeout}
        \State $(A, M) \gets$ pop lowest-cost Hamiltonian and gate sequence from $\mathcal{Q}$
        \For {each gate $G \in \mathcal{S}$}
            \If {$\mathcal{C}(GAG^\dagger) < \mathcal{C}(A)$}
                \State insert $(GAG^\dagger, M+[G])$ into $\mathcal{Q}$
                \If {$\mathcal{C}(GAG^\dagger) < b$}
                    \State update best cost $b \gets \mathcal{C}(GAG^\dagger)$ and best gate sequence $B \gets M$
                \EndIf
            \EndIf
        \EndFor
    \EndWhile
    \State \Return $B$
    \end{algorithmic}
\label{alg:SA}
\end{algorithm*}

\subsection{Performance comparison}\label{sec:comparison}

Our optimization protocol encompasses various degrees of freedom, allowing us to make flexible choices, such as selecting input Hamiltonians, gate sets, and algorithms. We conduct a succinct benchmarking process to determine the most effective combination for our specific optimization goals.

First, we discuss the choice of the search algorithm. In both best-first search and simulated annealing, we iteratively apply Clifford transformations to the Hamiltonian. The difference lies in how we choose the Clifford and how each iteration connects to the subsequent one. We observe that for small system sizes, both algorithms obtain the same optimized cost, whereas for larger systems, best-first search obtains a smaller improvement faster, but with a proper annealing schedule and sufficient time, the simulated annealing outperforms best-first search.

Next, both algorithms require an input qubit Hamiltonian. This is typically chosen from a conventional fermion-to-qubit mapping applied to the given fermionic Hamiltonian. For a given fermionic Hamiltonian, we find no discernible difference in optimized Pauli weight across initial maps beyond the randomness of the simulated annealing algorithm.

Lastly, we compare the performance of different gate sets $\mathcal{S}_{C}$, $\mathcal{S}_{CH}$, and $\mathcal{S}_{CHS}$. As to the choice of these gate sets, recall $\mathcal{S}_C$ consists only of $\CNOT$ gates. We would like to extend the search space to include other Clifford gates such as Hadamard and phase gates. However, a Hadamard will swap Paulis via ${\rm H}_i X_i {\rm H}_i = Z_i$, but it alone will not change the Pauli weight. Therefore we only include a Hadamard gate together with a CNOT gate with a shared qubit, and they are treated as a single unit.

We observe that $\mathcal{S}_{CH}$ has better optimization performance than $\mathcal{S}_{C}$, whereas we observe minimal difference in the performance of $\mathcal{S}_{CH}$ and $\mathcal{S}_{CHS}$.
This observation regarding the equivalent performance between $\mathcal{S}_{CH}$ and $\mathcal{S}_{CHS}$ is presented as a numerical conjecture. In addition, while the best result is the same for $\mathcal{S}_{CH}$ and $\mathcal{S}_{CHS}$, the behaviors during the algorithm are different. Adding $S$ gates slows down the optimization process. However, it may work to help escape from a local minimum, as it provides additional transformations of the Hamiltonian.

\subsection{Ternary-tree search}

For a small number of qubits $n$, it is possible to brute-force enumerate all possible trees on $n$ qubits. By doing so, we demonstrate Hamiltonians for which the optimized fermion-to-qubit map outperforms any ternary-tree mapping.

As formalized in Sec. \ref{sec:single-ops}, a ternary-tree mapping consists of three pieces of information: (1) its shape, (2) a labeling of parent nodes, and (3) a labeling of leaves. The labeling of parent nodes does not affect the Pauli weight of resulting terms, as they only affect the ordering of the qubits; swapping two parent node labels corresponds to swapping two qubits, i.e. a SWAP gate.

In general, the number of ternary-tree shapes on $n$ nodes is $\frac{1}{2n+1}\binom{3n}{n}$, which grows exponentially fast.\cite{graham1994concrete} However, we may consider two trees to be equivalent for the purpose of Pauli weight if they have the same set of children for each parent and differ only in the order of each parent's children. Permuting a given parent's left, middle, and right children corresponds to permuting the Pauli $X$, $Y$, and $Z$ for that qubit, which does not change the Pauli weight. Without loss of generality, we need only enumerate ternary-tree shapes with nodes flushed right. For $n=4$, there are $4$ such trees.

The main computational cost is to iterate over all $(2n+1)!$ enumerations of the leaves for each tree shape. This quickly becomes intractable for large $n$, but is feasible for $n=4$ with $4 \cdot 9! \approx 1.5 \times 10^6$ cases. For the hydrogen molecule with $2$ atoms (15 terms), we find that the best ternary-tree has total Pauli weight $32$ whereas our algorithm gives a mapping with total Pauli weight $26$.

We conjecture that this distinction is due to the presence of interaction terms. Even with a single exchange interaction term of the form $a_1^\dagger a_2^\dagger a_3 a_4 + h.c.$, we find that the best ternary-tree mapping obtains a total Pauli weight of $20$ whereas our algorithm gives a mapping with total Pauli weight $16$. However, we find no advantage present when we include only hopping terms or density interactions.

%%%%%%%%%%%%%%%%%%%%%%%%%%%%%
\section{Proofs}
\label{sec:proofs}

Here we give formal proofs of the results in Sec \ref{sec:single-ops}.

Given a ternary-tree, we will call the \emph{binary subtree} the subgraph induced by the root and all left and right children. Note the Jordan-Wigner and Bravyi-Kitaev trees both have all their parent nodes in their binary subtrees; we call such trees \emph{binary-shaped}.

\begin{lemma}
\label{lem:tree-rotation}
Let $T$ be a ternary-tree mapping where node $j$ is a left child of node $k$. Then conjugating by $\CNOT_{jk}$ results in a ternary-tree mapping where node $j$ is rotated right (see Figure \ref{fig:tree-rotation} in the main text) about node $k$ in the binary subtree of $T$, and all middle children retain their parent.
\end{lemma}

\begin{lemma}
\label{lem:tree-middle}

Let $T$ be a ternary-tree mapping where node $j$ is a middle child of node $k$. Then conjugating by $\CNOT_{jk}$ transforms the tree as in Figure \ref{fig:tree-middle} in the main text.
\end{lemma}

\begin{proof}[Proofs of Lemma \ref{lem:tree-rotation} and Lemma \ref{lem:tree-middle}]
Consider the $5$ children among nodes $j$ and $k$. For each child's subtree, all leaves in that subtree have the same Paulis on qubits $j$ and $k$. Since $\CNOT_{jk}$ only acts on those two qubits, the subtree will move as one entity, and it suffices to determine the action of $\CNOT_{jk}$ of Paulis restricted to those two qubits, as given in Table \ref{tab:cnot-actions}. The resulting tree is constructed to satisfy the $5$ transformed Paulis.
\end{proof}

\begin{proof}[Proof of Theorem \ref{thm:jw-bk}]
We construct a sequence of CNOTs that transforms the $n$-qubit Bravyi-Kitaev tree into the $n$-qubit Jordan-Wigner tree. Call the \emph{right spine} of the tree the root and all nodes which are right descendants from the root. The Jordan-Wigner tree has all parent nodes on the right spine. Using Lemma \ref{lem:tree-rotation}, we iteratively apply layers of CNOT gates to rotate nodes onto the right spine. 

On each iteration, iterate through every node $j$ that is a left child of a parent $k$ on the right spine, and apply $\CNOT_{jk}$. Each such CNOT gate increases the number of nodes on the right spine by $1$, so this procedure eventually terminates with every node on the right spine. Since both the Bravyi-Kitaev and Jordan-Wigner have both their parents and leaves labeled via an inorder traversal, and tree rotations maintain invariant the inorder traversal, we end with exactly the Jordan-Wigner mapping.
\end{proof}
\begin{remark}
By way of construction, we have also shown that the number of CNOTs $i_{\max}$ is less then $n$, and that the above procedure can be implemented in a CNOT circuit of depth $O(\log n)$.
\end{remark}

\begin{corollary}
Given any two $n$-qubit binary-shaped mappings with the same inorder traversal of nodes and the same inorder traversal of leaves, there exists a sequence of $O(n)$ CNOT gates transforming between them.
\end{corollary}

\begin{proof}
Using our above construction, we can find a sequence of CNOT gates $U_1$ that transforms the first mapping into a Jordan-Wigner-shaped tree, and likewise $U_2$ for the second. Then $U_2^\dagger U_1$ transforms the first into the second.
\end{proof}

\begin{proof}[Proof of Theorem \ref{thm:cnot-tree-shape}]
Similar to Theorem \ref{thm:jw-bk}, we iteratively transform each tree into the Jordan-Wigner right spine shape. For each node $k$ on the right spine, if it has a left child $j$, then we apply $\CNOT_{jk}$ to rotate node $j$ onto the spine by Lemma \ref{lem:tree-rotation}, increasing the number of nodes on the right spine by $1$. If node $k$ has a middle child $j$, then $\CNOT_{jk}$ moves node $j$ onto the right spine by Lemma \ref{lem:tree-middle}, also increasing the number of right spine nodes by $1$. This process terminates once all $n$ parent nodes are on the right spine.

Let $U_1$ be the sequence of Cliffords transforming $T_1$ to the Jordan-Wigner shape, and likewise $U_2$ for $T_2$. Then $U_2^\dagger U_1$ transforms $T_1$ into the shape of $T_2$.
\end{proof}

\begin{remark}
By way of construction, we have shown $i_{\max} \le 2n-2$ for this procedure.
\end{remark}

We remark that our algorithm does not perform the same sequence of gates as our proofs, but nevertheless finds an optimal-weight tree within the CNOT search space. 

\begin{proof}[Proof of Theorem \ref{thm:trees-included}]
First, using the construction in Theorem \ref{thm:cnot-tree-shape}, we convert both trees to the Jordan-Wigner shape, consisting solely of a right spine. Next, using SWAP gates, implemented as $\SWAP_{ij} = \CNOT_{ij} \CNOT_{ji} \CNOT_{ij}$, we obtain any permutation of qubit labelings. Now, for a given qubit $k$, the Hadamard and phase gates, $\mathrm{H}_k$ and $\mathrm{S}_k$, can achieve any permutation of leaves on node $k$. Finally, by Lemma \ref{lem:tree-rotation}, if $k$ is the right child of $j$, then $\CNOT_{jk} \mathrm{H}_j \CNOT_{jk}$ swaps the left child of node $j$ with the left child of node $k$. This combined with single-qubit gates can achieve any labeling of Majoranas. 
\end{proof}

\begin{table*}[h]
\centering
\begin{tabular}{|l|c|c|c|c|c|c|}
\hline
   & CNOT$_{i,j}$ 
   & CNOT$_{j,i}$
   &CNOT$_{i,j}$H${_i}$ &CNOT$_{j,i}$H${_i}$ &CNOT$_{i,j}$S${_i}$ &CNOT$_{j,i}$S${_i}$
\\
\hline
II & II& II& II& II& II& II\\ 
IX & IX& XX& IX& XX& IX& XX\\ 
IY & ZY& XY& ZY& XY& ZY& XY\\ 
IZ & ZZ& IZ& ZZ& IZ& ZZ& IZ\\ 
XI & XX& XI& ZI& ZZ& YX& ZY\\ 
XX & XI& IX& ZX& YY& YI& ZY\\ 
XY & YZ& IY& IY& YX& XZ& ZX\\ 
XZ & YY& XZ& IZ& ZI& XY& YI\\ 
YI & YX& YZ& YX& YZ& XX& XI\\ 
YX & YI& ZY& YI& ZY& XI& IX\\ 
YY & XZ& ZX& XZ& ZX& YZ& IY\\ 
YZ & XY& YI& XY& YI& YY& XZ\\ 
ZI & ZI& ZZ& XX& XI& ZI& ZZ\\ 
ZX & ZX& YY& XI& IX& ZX& YY\\ 
ZY & IY& YX& YZ& IY& IY& YX\\ 
ZZ & IZ& ZI& YY& XZ& IZ& ZI\\ 
\hline
\end{tabular}
\caption{Full table for the actions on two-qubit Pauli operators by Clifford gates, not including sign. The 2-qubit Pauli operations in the first column is transformed by the gates on the first row to produce the output 2-qubit Paulis on the rest of the column. The default order of the two-qubit Pauli is $(i,j)$.}
\label{tab:cnot-actions}
\end{table*}

\end{document}